\documentclass[journal,10pt]{IEEEtran}
\usepackage{amsmath,color,amssymb,dsfont}
\usepackage{amsthm}

\def\R{{\mathbb{R}}}
\def\Z{{\mathbb{Z}}}
\def\xx{{\boldsymbol{x}}}
\def\yy{{\boldsymbol{y}}}
\def\zz{{\boldsymbol{z}}}
\def\vv{{\boldsymbol{v}}}
\def\qq{{\boldsymbol{q}}}
\def\ee{{\boldsymbol{e}}}
\def\DD{{\boldsymbol{D}}}
\def\AA{{\boldsymbol{A}}}

\def\XX{{\boldsymbol{X}}}
\def\II{{\boldsymbol{I}}}
\def\zero{{\boldsymbol{0}}}

\def\argmax{\mathop{\rm arg\,max}}

\def\QQ{{\boldsymbol{Q}}}
\def\MM{{\boldsymbol{M}}}
\def\HH{{\boldsymbol{H}}}

\newtheorem{thma}{Theorem}

\newtheorem{corr}[thma]{Corollary}
\newtheorem{lem}[thma]{Lemma}

\newcommand{\bb}{0.199}
\usepackage[pdftex]{graphicx}

\newcommand{\Ce}{\mathbb{C}}

\begin{document}

\title{Compressive Shift Retrieval}

\author{
Henrik Ohlsson,~\IEEEmembership{Member,~IEEE,}  
Yonina C. Eldar,~\IEEEmembership{Fellow,~IEEE},  
Allen Y. Yang,~\IEEEmembership{Senior Member,~IEEE}, \\ and 
S. Shankar Sastry,~\IEEEmembership{Fellow,~IEEE} 
\thanks{Ohlsson, Yang and Sastry are with the Dept. of Electrical Engineering and Computer
   Sciences, University of California, Berkeley, CA, USA,
 e-mail: ohlsson@eecs.berkeley.edu.}
\thanks{Ohlsson is also with Dept. of Electrical Engineering, Link\"oping University,
SE-581 83 Link\"oping, Sweden.}
\thanks{Eldar is with the Dept. of Electrical Engineering, Technion --
  Israel Institute of Technology, Haifa 32000, Israel.}
\thanks{Ohlsson is partially supported by the Swedish Research
  Council in the Linnaeus center CADICS, the European Research Council
   under the advanced grant LEARN, contract 267381, by  a postdoctoral grant from the Sweden-America
   Foundation, donated by ASEA's Fellowship Fund, and by a postdoctoral
   grant from the Swedish Research Council. Eldar is supported in part
   by the Israel Science Foundation under Grant no. 170/10, and by the
   Ollendorf Foundation. Yang is supported in part by ARO 63092-MA-II, DARPA FA8650-11-1-7153 and ONR N00014-13-1-0341.}
\thanks{This paper was presented in part at the 38th International
  Conference on Acoustics, Speech, and Signal Processing (ICASSP), May
  26-31, 2013, \cite{Ohlsson13cc}. See also \cite{ohlsson13c}.} 
}


\maketitle

\begin{abstract}
The classical shift retrieval problem considers two signals in vector
form that are related by a shift. The problem is of great
importance in  many applications and is typically solved by maximizing the
cross-correlation between the two signals.  Inspired by compressive
sensing, in this paper, we seek to estimate the shift directly from
compressed signals.  We show that under certain conditions, the shift
can be recovered  using fewer  samples and less computation
compared to the classical setup.  
Of particular interest is shift
estimation  from Fourier  coefficients. We show that
under rather mild conditions only one Fourier  coefficient  suffices
to recover the true  shift.

\end{abstract}

\IEEEpeerreviewmaketitle

\section{Introduction}
\label{sec:intro}

\IEEEPARstart{S}{hift} retrieval is a fundamental
problem in many signal processing
applications. 
For example, to map the ocean floor, an active sonar can be used. The
sonar  transmits certain sound pulse patterns in the water,
and the time it takes to receive the echoes of the pulses indicates
the depth of the ocean floor. 
In target tracking using two acoustic sensors,  
%
the time shift when a 
sound wave of a vehicle reaches the microphones indicates the direction to the
vehicle. 
In the case of a time shift, the shift retrieval problem is often referred
to as \emph{time delay estimation} (TDE) \cite{CarterG1976}. In computer vision, 
the spatial shift  relating two images is often sought and 
referred to as
image registration or alignment \cite{LucasB1981,HagerG1998-PAMI,PengY2010}.

Traditionally, the shift retrieval problem is solved by maximizing the cross-correlation between the two
signals \cite{Hero98}.
In this paper, we revisit
this classical problem, and show how the basic premise
of \textit{compressive sensing} (CS)
\cite{CandesE2005-IT,Candes:06,Donoho:06,EldarY2012} can be used in
the context of shift retrieval. This allows to recover the shift from
compressed data leading to computational and storage savings. 

Compressive sensing is a sampling scheme that makes it possible to sample at the
\emph{information rate} instead of the classical Nyquist rate predicted by the
bandwidth of the signal \cite{Mishali:11}. The majority of the results
in compressive sensing discuss conditions and methods for guaranteed
reconstruction from an under-sampled version of the signal. Therefore,
the information rate is typically referred to as the one that guarantees the recovery of the sparse signal.

However, for many applications such as the aforementioned examples in
shift retrieval, obtaining the signal may not be needed. The goal is to recover some properties or statistics of the unknown signal.
Taking the active sonar for example, one may wonder if it is really necessary to sample at a
rate which is twice that of the bandwidth of the transmitted signal so
 that the received signal can be exactly reconstructed? Clearly the answer is no.
Since the signal itself is not of interest to the application, we might consider an alternative sampling scheme to directly
estimate the shift without first reconstructing the signal. 
These ideas have in fact been recently explored in the context of
radar and ultrasound \cite{Chernyakova13,BarIlan12,Baransky12,Vetterli02} with continuous time signals and multiple shifts.
Here we consider a related problem and ask: \emph{What is the minimal information rate to shift retrieval 
when two related discrete-time signals are under-sampled?}

It turns out that under rather mild conditions, we only need fractions
of the signals. In fact, we will show that 
only one Fourier coefficient from each of the signals suffices
to recover the true  shift. We refer to
the method as \textit{compressive shift retrieval} (CSR).
It should be made clear that CSR does not assume that any of the
involved signals are necessarily sparse.

As the main contribution of the paper, we will  show that when
the sensing matrix  is taken to be a
partial Fourier matrix, under suitable conditions, the true shift can be recovered from both
noise-free and noisy measurements using CSR.  Furthermore,
CSR reduces both the computational load and the number of samples
needed in the process. This is of particular interest since recent developments
 in sampling \cite{Tur:2011,Gedalyahu:11,Baransky:12} have shown that Fourier coefficients can 
be efficiently obtained from space (or
time) measurements by the use of  an appropriate filter and by subsampling the
output. 
Remarkably, our results also show that in some cases sampling as few
as one Fourier
coefficient is enough to
perfectly recover the true shift. 



\subsection{Prior Work}
\label{sec:prior}

Compressive signal alignment
problems have been addressed in 
only a few publications and, to the authors' best knowledge, not in
the same setup studied in this paper. In 
\cite{KokiopoulouE2009}, the authors considered alignment of images
under random projection. The work was based on the
Johnson-Lindenstrauss property of random projection and proposed an
objective function that can be solved efficiently using
\emph{difference-of-two-convex programming} algorithms. 
In this paper, we instead focus on proving theoretical guarantees of
exact shift recovery when the signal is subsampled by a partial
Fourier basis. The theory developed in  \cite{KokiopoulouE2009}  does
not apply to this setup. 
The \emph{smashed filter} \cite{davenport2007smashed} is another related technique. It is a general framework for maximum likelihood hypothesis
testing and can be seen as a dimensionally reduced matched filter. It
can therefore be applied to the shift retrieval problem. The
underlying idea of both the smashed filter and CSR are the same in
that both approaches try to avoid reconstructing the signal and
extract the sought descriptor, namely, the shift, from compressive
measurements. However, the analysis and assumptions are very
different. For CSR, we  develop requirements for guaranteed
recovery of the true shift for a given measurement matrix. For the smashed filter, the analysis focuses on random orthoprojections and provides
probabilities for correct recovery as a function of the number of
projections. Also, in this paper, we are particularly interested in Fourier
measurements, and many of the results are therefore tailored to
this setting.  The work presented here can therefore be seen
as complementary to what was presented in \cite{davenport2007smashed}
and its extension in \cite{davenport2010signal}.


\subsection{Notation and Assumptions}
We use normal fonts to represent scalars and bold fonts for vectors and
matrices. The notation $|\cdot|$ represents the absolute value for scalars, vectors
and matrices, and it returns the cardinality of a set if the
argument is a set.  For both vectors and matrices, $\|\cdot \|_0$ is the $\ell_0$-norm function that returns the
number of nonzero elements of its argument. Similarly, $\|\cdot\|_p$ represents the
$\ell_p$-norm. For a vector $\xx$, the $\ell_p$-norm is defined as $\|{\xx}\|_p \triangleq  (\sum_i |x_{i}|^p)^{1/p}$, where $x_i$ is the $i$th element of $\xx$.  For
a matrix $\XX$,  $\|\cdot\|_p$ is defined as
$ {\XX} \triangleq  (\sum_{i,j}| { X}_{i,j}|^p)^{1/p}$, where ${X }_{i,j}$ is the $(i,j)$-th element of $\XX$.
Furthermore, $\XX^*$ denotes the complex conjugate
transpose of $\XX$. Let 
$\II_{n\times n}$ denote an $n\times n$ identity matrix, $\zero_{m \times
    n}$ an $m \times n$ matrix of zeros,
and $\mathbb{Z}$ be the set of integers. $\Re \{\cdot \}$ returns
the real part of its argument. We say that two
$n$-dimensional vectors $\yy$ and $\xx$ are related
by an $l$ cyclic-shift if
$\yy=\DD^l \xx$, where $\DD^l$ is defined as
\begin{equation}
\DD^l=\begin{bmatrix} \zero_{l \times (n-l) } & \II_{l \times
    l} \\ \II_{(n-l) \times (n-l) } & \zero_{(n-l) \times
    l} 
\end{bmatrix}.
\end{equation}
Throughout the paper, we will
assume that the shift is unique  \emph{up to a
multiple of $n$}. Also note that we are considering cyclic shifts.


\subsection{Organization}

In Sections \ref{nfCSR} and \ref{sec:Fourier}, we study the CSR
problem under the assumption that the measurements are noise free. Next, we
extend the results to noisy measurements in Section \ref{nCSR}. As we are particularly
interested in Fourier measurements, we will
 tailor the results to this particular choice of sensing matrix. 
Section \ref{sec:con} concludes the paper. All proofs are provided in the
 Appendice for clarity. 

\section{Noise-Free Compressive Shift Retrieval}\label{nfCSR} 

The shift retrieval problem is a multi-hypothesis testing problem:
Define the $s$th hypothesis $\mathcal{H}_s$, $s=0,\dots,n-1$, as
\begin{align*}
\mathcal{H}_s:\: \xx \text{ is related to } &\yy \text{ via a }
 s\text{-cyclic-shift},
\end{align*}
and accept $\mathcal{H}_{s}$ if
$\yy=\DD^s\xx$ and otherwise reject. Since the true shift is assumed unique, only one
hypothesis will be accepted and the corresponding shift  is
necessarily  the true solution. The connection to the cross-correlation is now trivial:
\begin{align}\nonumber
\|\yy- \DD^s\xx\|_2^2 =&\|\yy\|^2_2+ \|\DD^s\xx\|_2^2 -\yy^* \DD^s\xx-
\DD^s\xx^*\yy \\=&\|\yy\|^2_2+ \|\xx\|_2^2 -2\Re \{ \langle \yy,   \DD^s  \xx\rangle \}  
\end{align} where we use the fact that $\|\DD^s \xx\|_2^2=\|\xx\|_2^2$. Since
$\|\yy- \DD^s\xx\|_2^2 \geq 0$, equating $\yy=\DD^s\xx$ is equivalent
to minimizing  $\|\yy- \DD^s\xx\|_2^2$ or maximizing the real part
of the cross-correlation with respect to $s$:
\begin{equation}
\max \Re \{ \langle \yy,  \DD^s  \xx\rangle \}.
\label{eq:test2}
\end{equation}

Now, assume that the compressed measurement signals $\zz=\AA\yy \in \Ce^m$ and $\vv=\AA\xx \in \Ce^m$ are given and related to
the ground-truth signals $\xx\in \Ce^n$ and its shifted version
$\yy=\DD^l \xx\in \Ce^n$ via the sensing matrix $\AA\in
\Ce^{m\times n}$, $m<n$. The goal of CSR
is to recover  the shift $l$ relating $\xx$ and $\yy$ from the
compressed measurements $\zz$ and $\vv$.

Since only the compressed
measurements $\zz$ and $\vv$ are assumed available, we can not
evaluate $\yy=\DD^s\xx$ or maximize $\Re \{ \langle
\yy,   \DD^s  \xx \rangle \}$ for each hypothesis $\DD^s$. However, if $\AA^*\AA$ and $\DD^s$ commute for all
$s=0,\dots, n-1$, then 
\begin{align}\nonumber 
\yy=\DD^s\xx \quad  \Rightarrow &\quad \AA^*\AA \yy=\AA^*\AA\DD^s\xx =\DD^s
\AA^*\AA\xx \\ \Leftrightarrow & \quad \AA^*\zz=\DD^s\AA^*\vv. 
\end{align}
Hence, we could consider the test:  
\begin{equation}
\text{Accept } \mathcal{H}_{s} \text{ if
} \AA^*\zz=\DD^s\AA^*\vv \text{ and otherwise reject}.
\label{eq:testing}
\end{equation}
It is clear that if $s$ is such that $\yy=\DD^s\xx$, then
$\AA^*\zz=\DD^s\AA^*\vv$ will also hold. However, the other way around might
not be true. Therefore, we might erroneously accept a wrong hypotheses using \eqref{eq:testing}. The next theorem lists the conditions by which the testing \eqref{eq:testing} is guaranteed to accept the correct hypothesis. Notice that testing the condition  $
\AA^*\zz=\DD^s\AA^*\vv $ is equivalent to minimizing $
\|\AA^*\zz-\DD^s\AA^*\vv \|_2^2$ with respect to $s$. 

 \begin{thma}[\bf Shift Recovery from Low-Rate Data]\label{thm:first}
Let $\XX$ be an $n\times n$ matrix
 with the $i$th column equal to $\DD^i \xx$, $i=1,\dots,n,$ and define
  $\bar \DD^s = \AA  \DD^s \AA^*$.
If the sensing matrix $\AA$ satisfies the following conditions:
\begin{itemize}
\item[1)]  $\AA^* \AA \DD^s = \DD^s  \AA^* \AA$,  
\item[2)]  $\exists \alpha\in \R, \alpha \AA \AA^* = \II$ and 
\item[3)] all columns of $\AA\XX$  are different, 
\end{itemize}
then 
\begin{equation}\label{eq:prodtest4}
 \max_{s} 
\Re \{ \langle \zz,  \bar \DD^s  \vv\rangle \}  
 \end{equation} or equivalently the test \eqref{eq:testing} recovers the true shift.
\end{thma}
The conditions of Theorem \ref{thm:first} may seem
restrictive. However, as we will show in Lemma \ref{lem:com}, if $\AA$ is
chosen as a partial Fourier matrix, then  the first two conditions of
Theorem \ref{thm:first}  are trivially satisfied. The last condition
is the only one that needs to be checked and will lead to a condition
on the sampled Fourier coefficients. 

The conditions of Theorem \ref{thm:first} can be checked prior to estimating the shift. However, knowing the
estimate of the shift, it is easy to see from the proof (see the proof
of Lemma
\ref{lem:first}) that it is enough to check 
 if the column of $\AA\XX$ associated
with the estimate of the shift is different than all the other
columns of $\AA \XX$. Hence, we do not need to check if all columns of $\AA\XX$
are different. This conclusion is formulated in the following
corollary, which is less conservative than Theorem \ref{thm:first}.
\begin{corr}[\bf Test for True Shift]\label{cor:test}
Let $\XX$ be an $n\times n$ matrix
 with the $i$th column equal to $\DD^i \xx$, $i=1,\dots,n,$ and define
  $\bar \DD^s = \AA  \DD^s \AA^*$.
If the sensing matrix $\AA$ satisfies the following conditions:
\begin{itemize}
\item[1)]  $\AA^* \AA \DD^s = \DD^s  \AA^* \AA$,  and
\item[2)]  $\exists \alpha\in \R, \alpha \AA \AA^* = \II$,  
\end{itemize}
then 
\begin{equation}
 s^* =\argmax_{s} 
\Re \{ \langle \zz,  \bar \DD^s  \vv\rangle \}  
 \end{equation} is the true shift if the $s^*$th column of $\AA\XX$ is
 different than all the other columns of  $\AA\XX$.
\end{corr}

\section{Compressive Shift Retrieval using Fourier Coefficients}
\label{sec:Fourier}
Of particular interest is the case where $\AA$ is made up of a partial Fourier
basis. That is, $\AA$ takes the form
\begin{equation*}
\AA=\frac{1}{\sqrt{n}}\begin{bmatrix}
1& e^{- \frac{2j \pi k_1 }{n} } &
  e^{- \frac{4j \pi k_1 }{n} }& \cdots & e^{-
    \frac{2(n-1) j  \pi k_1 }{n} } \\1& e^{- \frac{2j \pi k_2
    }{n} }& \ddots &&  e^{- \frac{2(n-1) j
    \pi k_2  }{n} } \\ \vdots &\vdots&& & \\1&
e^{- \frac{2j \pi k_m }{n} }& e^{- \frac{4j
    \pi k_m }{n} } &\cdots & e^{-
  \frac{2(n-1) j
    \pi k_m  }{n} }
\end{bmatrix}
\end{equation*}
 where $k_1,\dots, k_m \in \{0,1,2,\dots n-1\}, m \leq n$.  For this
 specific choice, 
\begin{equation*}
\AA\XX =\frac{1}{\sqrt{n}}\begin{bmatrix} X_{k_1} & X_{k_1} e^{\frac{2k_1 \pi j}{n}}& \cdots
  &  X_{k_1} e^{\frac{2(n-1)k_1 \pi j}{n}}  \\  X_{k_2} & \ddots &&
  X_{k_2} e^{\frac{2(n-1) k_2\pi j}{n}}  \\ \vdots&& & \\ X_{k_m} &
  X_{k_m} e^{\frac{2 k_m \pi
    j}{n}}  &\cdots &  X_{k_m} e^{\frac{2(n-1) k_m \pi j}{n}}  
\end{bmatrix}
\end{equation*}
where $X_r$ denotes the $r$th Fourier coefficient of the
Fourier transform of $\xx$.

For a sensing matrix made up by a partial Fourier basis, we have the
following useful result:
\begin{lem} 
\label{lem:com}
 Let $\AA $ be a partial Fourier matrix. Then $\DD^s \AA^*\AA = \AA^* \AA \DD^s$ for
 all $s = 1,\dots, n$.
\end{lem}
Using this result in Theorem 1 gives the following corollary:
\begin{corr}[\bf Shift Recovery from Low Rate Fourier Data]\label{cor:first}
With $\AA$  denoting a partial Fourier matrix and  $z_i$ and $v_i$  the $i$th element
of $\zz$ and $\vv$,
\begin{equation}\label{eq:simptest}
\max_{s} \Re \left \{ \sum_{i=1}^m z_i v_i e^{\frac{-2 \pi j k_i s}{n}}
\right \}
\end{equation}
recovers the true shift if there exists $p\in \{1,\dots,m \}$ such
that $X_{k_p}\neq 0$ and $\{1,\dots,n-1 \} \frac{k_p}{n}$ contains no
integers. In particular, measuring  only the first Fourier
coefficients ($k_1=1$) of $\xx$ and $\yy$ would, as long as the
coefficients are nonzero, suffice  to
recover the true shift.
\end{corr}
Remarkably,  in the extreme case when $m=1$, the corollary states that all we need is two scalar measurements,
$z$ and $v$, to perfectly recover the true shift. The scalar
measurements can be any nonzero Fourier coefficient of $\xx $ and $\yy$ as
long as $\{1,\dots,n-1 \} \frac{k_1}{n}$ contains no
integers. As noted in the corollary, the first Fourier coefficients
($k_1=1$) of $\xx$ and $\yy$  would suffice.  Also note that only $2mn$ multiplications are  required to
evaluate the test. This should be compared with $n^2$ multiplications
to evaluate the cross-correlation for the full uncompressed signals $\xx$ and $\yy$.
 Corollary~\ref{cor:first} is easy to check but  more conservative than both Theorem
\ref{thm:first} and Corollary \ref{cor:test}.

To validate the results, we carried out the following example. In each
trial we let the sample dimension $m$ and the shift $l$ be random integers between 1 and 9 and  generate $\xx$ by sampling from a
$n$-dimensional  uniform distribution. We let $n=10$ and make sure
that $\AA $ in each trial is a
partial Fourier basis satisfying the assumptions of  Corollary \ref{cor:first}. We carry out 10000
trials. 
The true shift is successfully recovered in each trial by the simplified test
\eqref{eq:simptest}, namely, with 100\% success rate. This is
quite remarkable since when $m=1$, we recover the true shift using only
two scalar measurement $\zz$ and $\vv$ and $1/5$ of the multiplications that
maximizing the real part of inner product between the original signals
\eqref{eq:test2} would need.  

\section{Noisy Compressive Shift Retrieval}\label{nCSR} 
Now we consider the noisy version of compressive shift retrieval, where the measurements $\zz$ and $\vv$ are perturbed
by  noise:
\begin{align}
\tilde \zz = \zz+ \ee_z,\quad 
\tilde \vv = \vv+ \ee_v.
\end{align}
Similar to the noise-free case, here we can also guarantee the recovery of
the true shift. Our main result is given in the following theorem: 
\begin{thma}[\bf Noisy Shift Recovery from Low-Rate Data]\label{thm:noisyrecov} 
Let $\tilde \xx$ be such that $\tilde \vv=\AA 
  \tilde \xx$, the $i$th column of $ \tilde \XX $ be  shifted
  versions of $\tilde \xx$,  and assume that $\AA $ is a partial Fourier
  matrix and that the noisy measurements are used in \eqref{eq:simptest} to estimate the shift. If
the 
$\ell_2$-norm difference between any two columns of $\AA  \tilde \XX $ is greater
than {\small
 \begin{equation}\nonumber 
 \Delta_{\zz \vv} \triangleq
\| \ee_{z}\|_2+\| \ee_v \|_2+\sqrt{ \|\tilde
     \vv\|_2^2+\|\tilde \zz\|_2^2-2 \max_s \Re \{ \langle \tilde \zz,
     \bar \DD^s  \tilde \vv\rangle \}   },
 \end{equation} }
then the estimate of the shift  is not affected by the noise. 
\end{thma}
Note that the theorem only states that the noise does not affect the
estimate of the shift. It does not state that the shift will be the
true shift.

We illustrate the results 
by running a Monte Carlo simulation
consisting of 10000 trials for each sample dimension $m=1,\dots,10,$ and for two different
SNR levels. 
In Figure~\ref{fig:highSNR}, 10 histograms are shown (corresponding to
$m=1,\dots,10$) for the
\begin{equation}
\text{SNR}=\frac{\|\zz\|_2 ^2 }{ \|\tilde \zz-\zz\|_2 ^2}  
\end{equation}
being 2 (low SNR) and in Figure~\ref{fig:lowSNR}, $SNR=10$ (high SNR).
The errors $\ee_z$ and $\ee_x$ were both generated by sampling from
\begin{equation}
\mathcal{N}(0,\sigma^2)+j \mathcal{N}(0,\sigma^2).
\end{equation}  
We further use $n=10$, $l=5$ and sample $\xx$  from a uniform (0,1)-distribution. The conclusion from the simulation is that the smaller
the $m$, the more the estimate of the shift is sensitive to
noise. Notice that when $m=10$ the test \eqref{eq:simptest} reduces to
the classical test of maximizing the cross-correlation.

 We
can now use Theorem \ref{thm:noisyrecov} to check if the noise
affected the estimate of the shift or not in each of the trials. For
$m=2$ and the high SNR, 40\% of the trials satisfied the conditions of
Theorem \ref{thm:noisyrecov}  and the noise therefore did not affect
the shift estimates in those cases. Of the trails that satisfied the
conditions, all predicted the true shift and none a false shift. Note
however that Theorem \ref{thm:noisyrecov} only states that if the
conditions are satisfied, then the estimated shift is  the same
as if we would have used the noise free compressed
measurements in the test \eqref{eq:simptest}. It does not state that
the estimate will be the true shift.

\begin{figure}[h!]
\includegraphics[width=\bb\columnwidth]{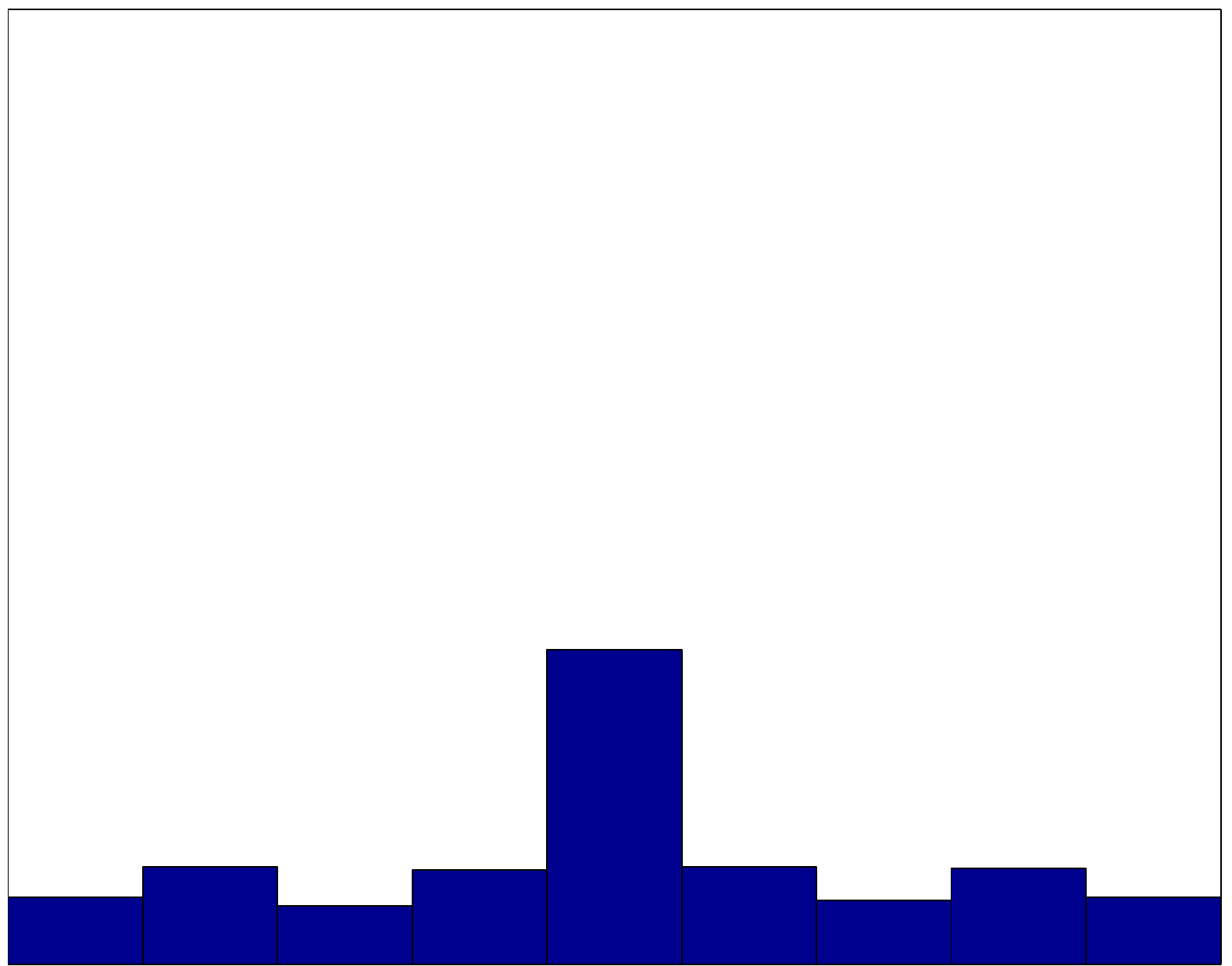}\includegraphics[width=\bb\columnwidth]{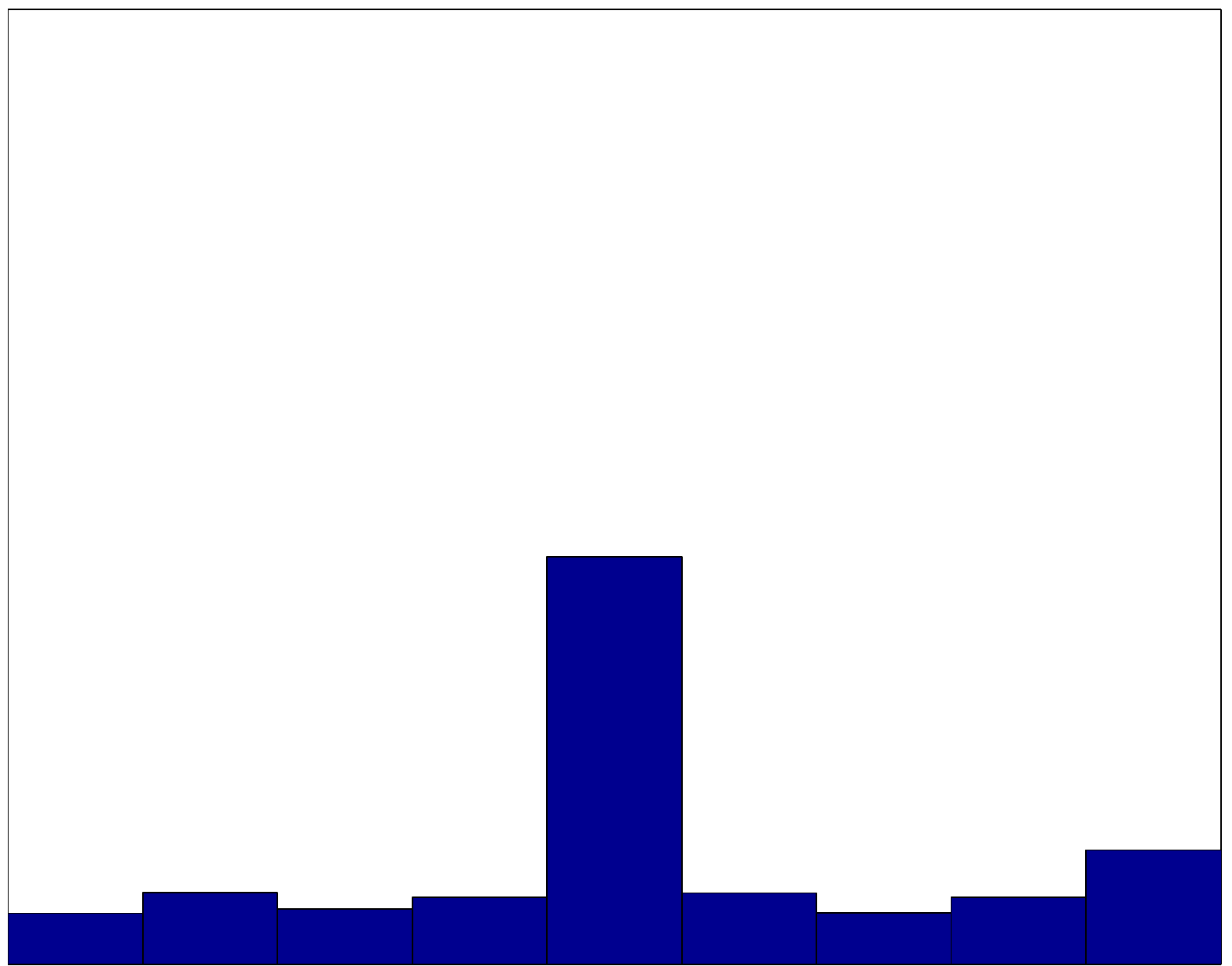}\includegraphics[width=\bb\columnwidth]{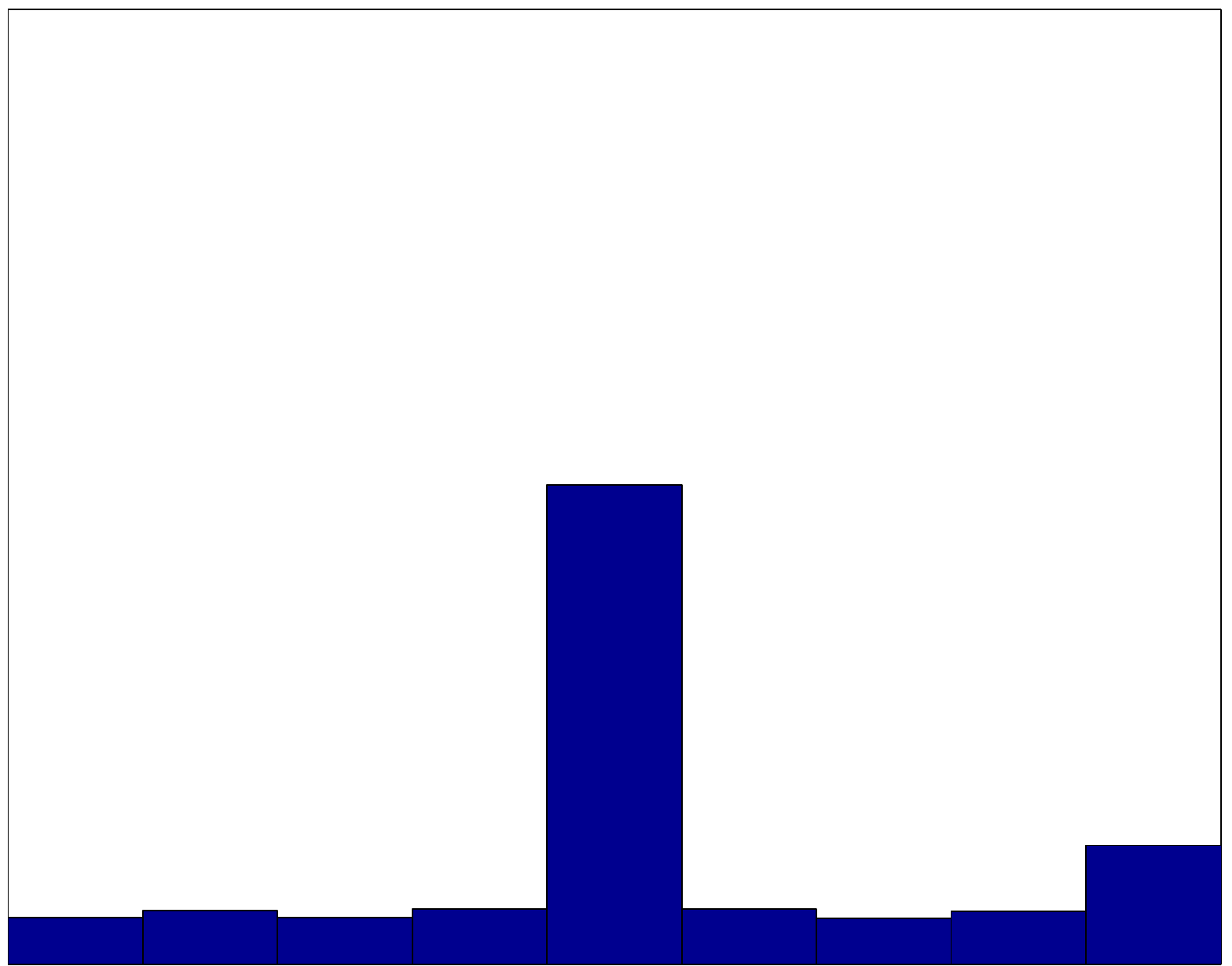}\includegraphics[width=\bb\columnwidth]{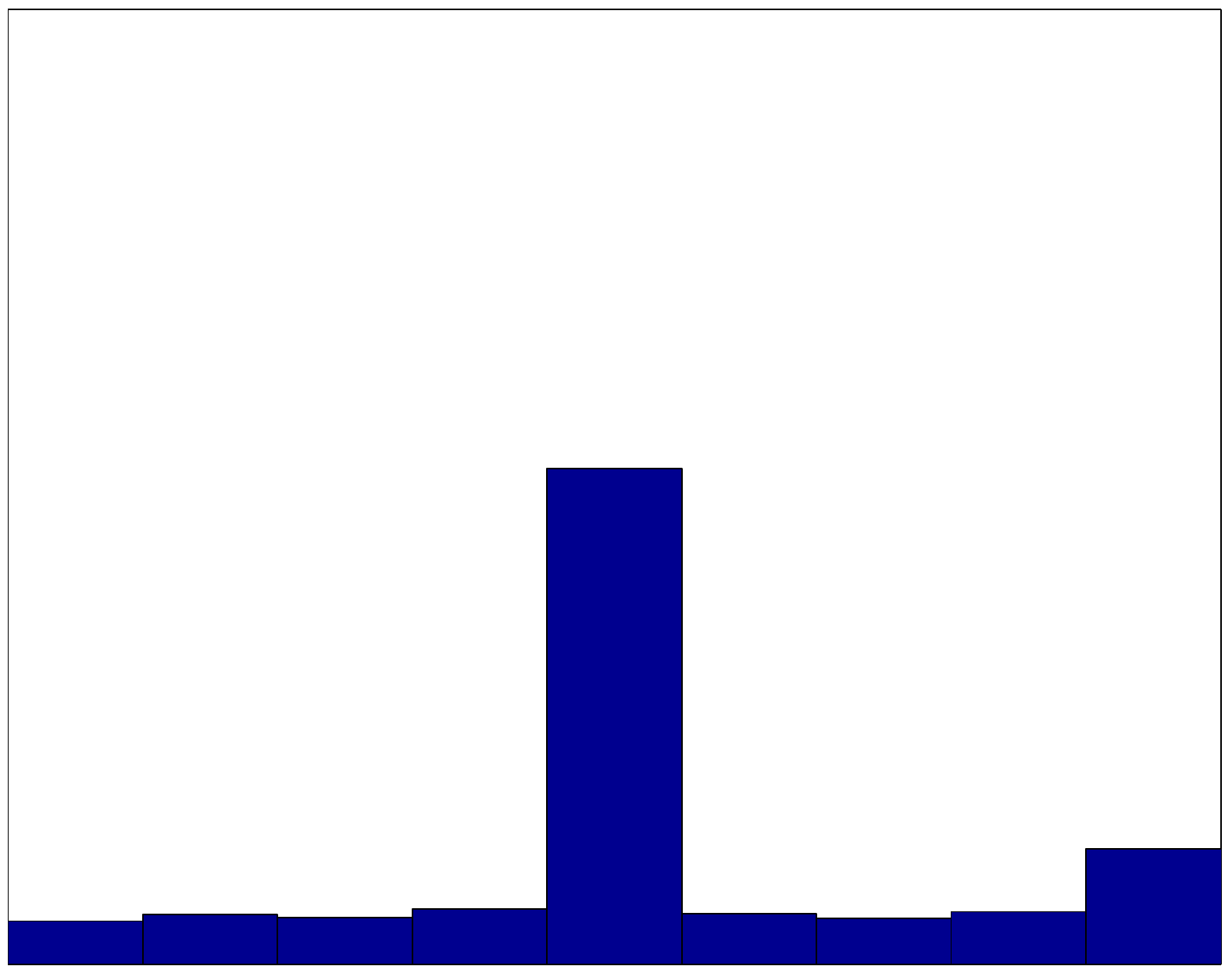}\includegraphics[width=\bb\columnwidth]{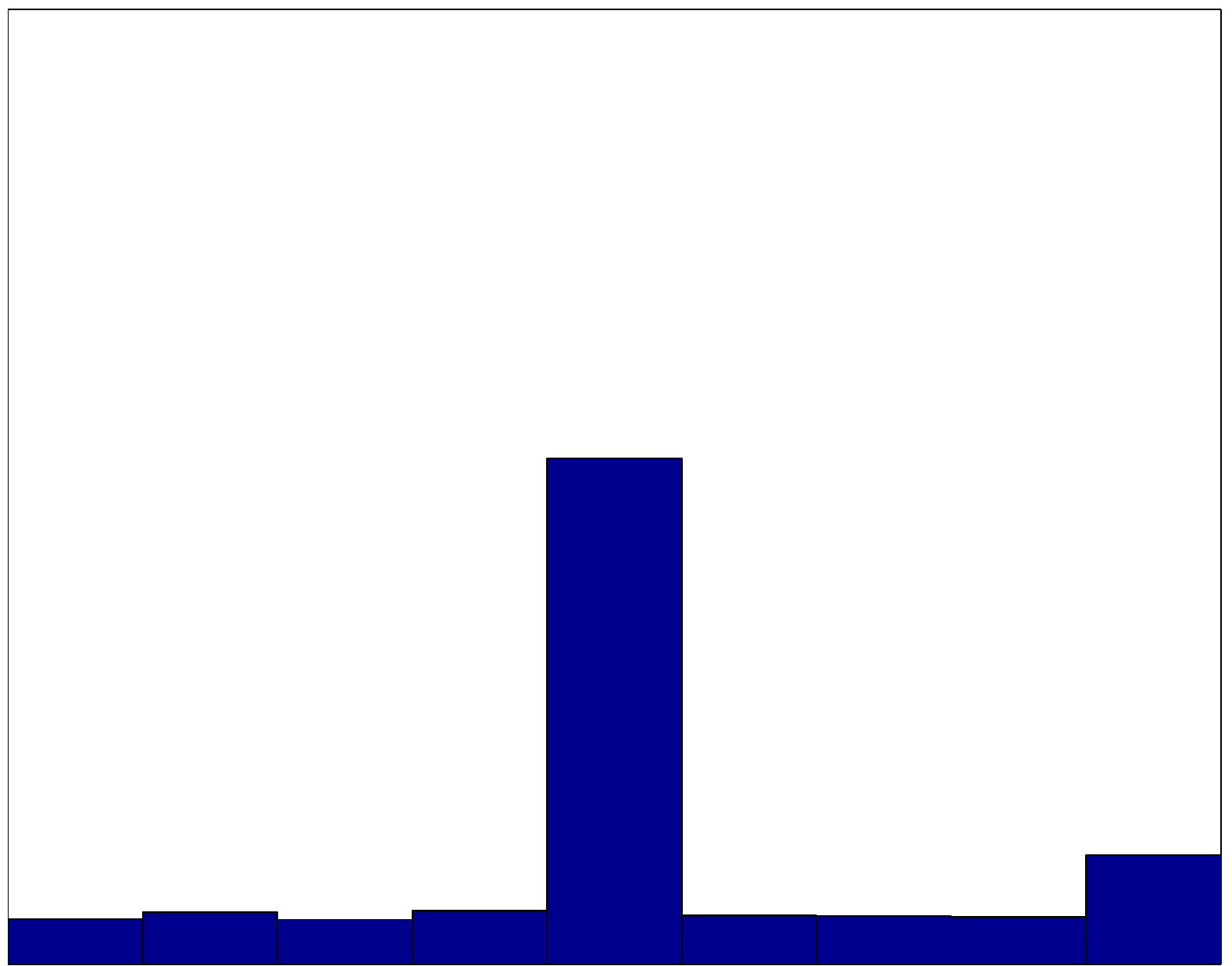}\\\includegraphics[width=\bb\columnwidth]{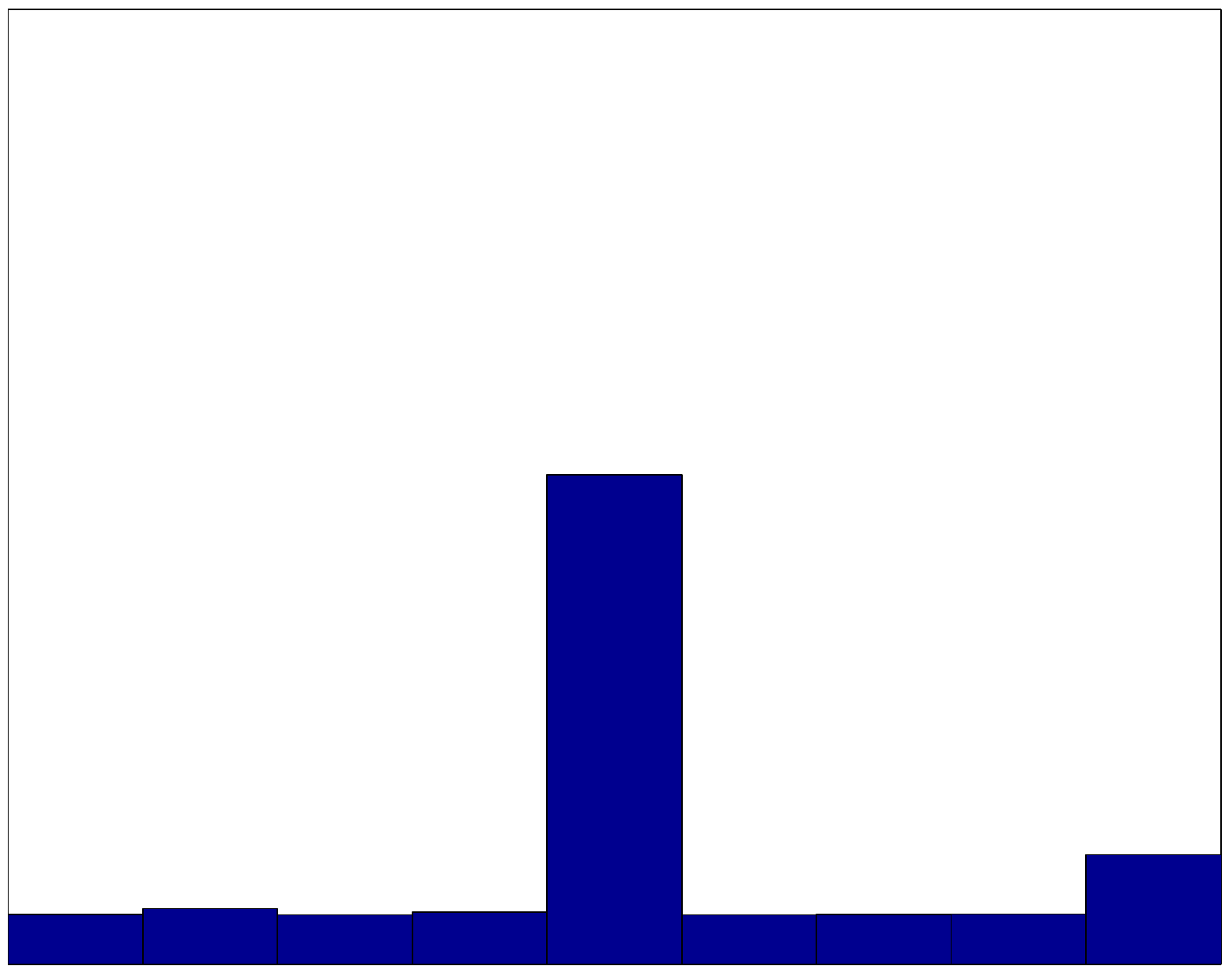}\includegraphics[width=\bb\columnwidth]{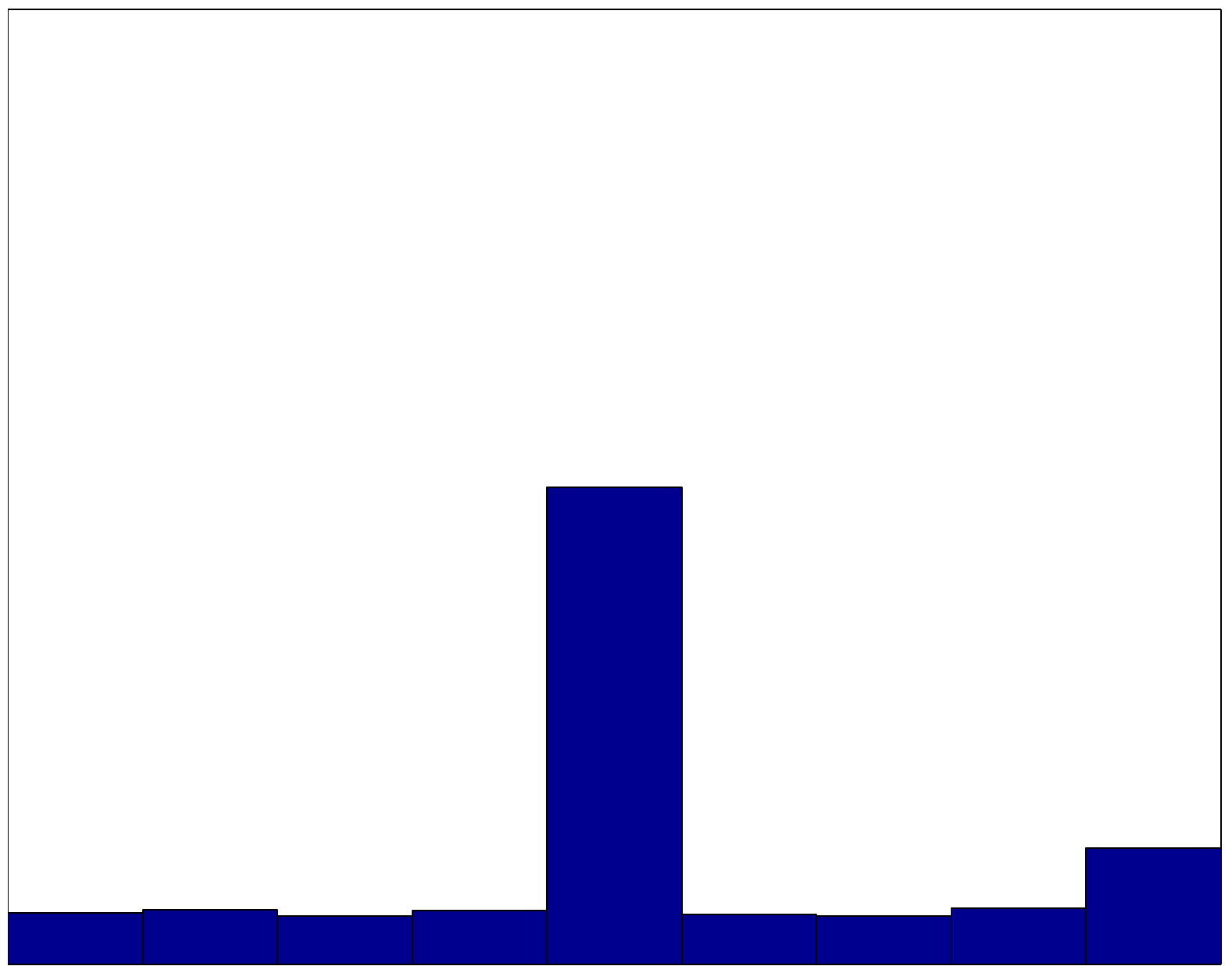}\includegraphics[width=\bb\columnwidth]{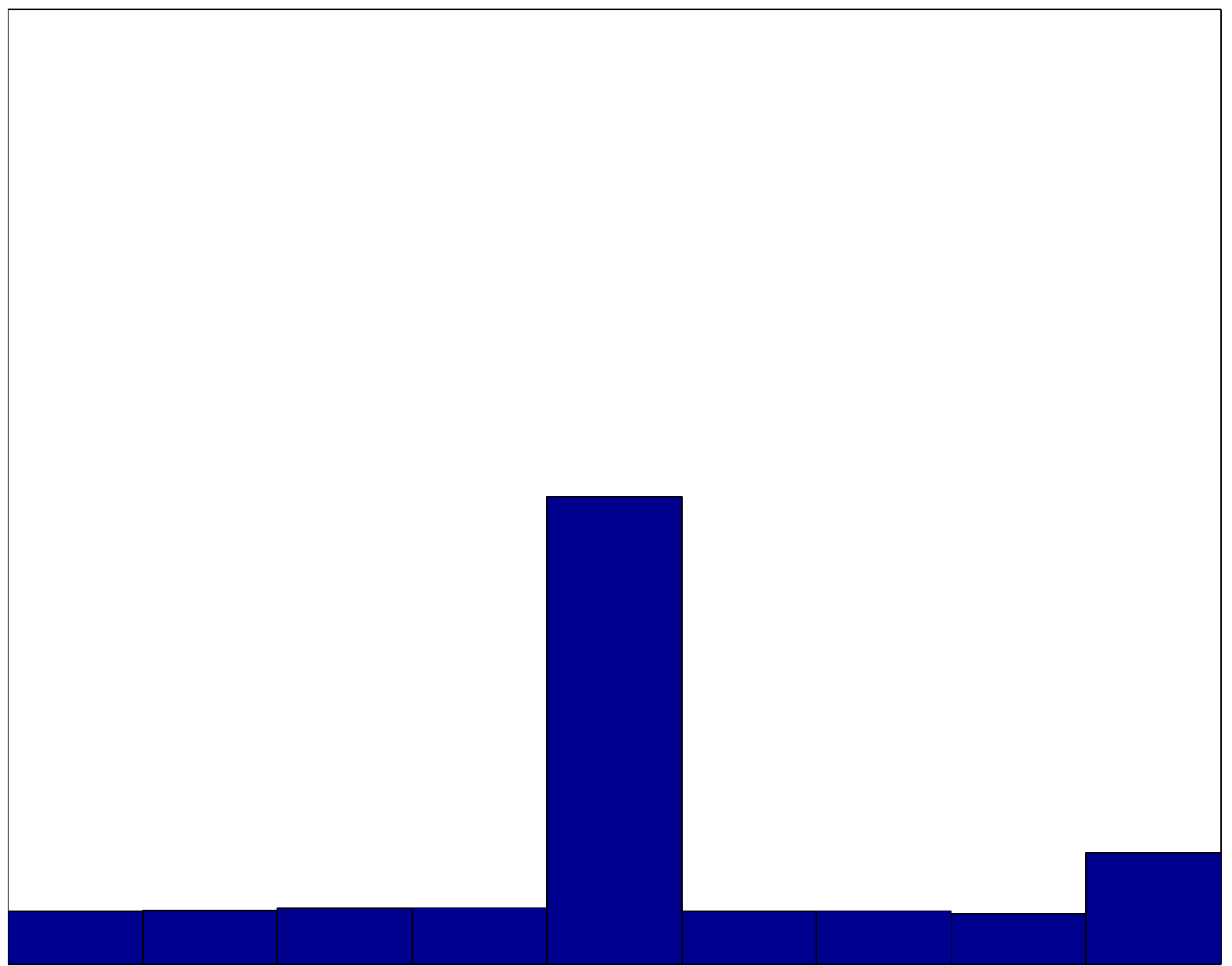}\includegraphics[width=\bb\columnwidth]{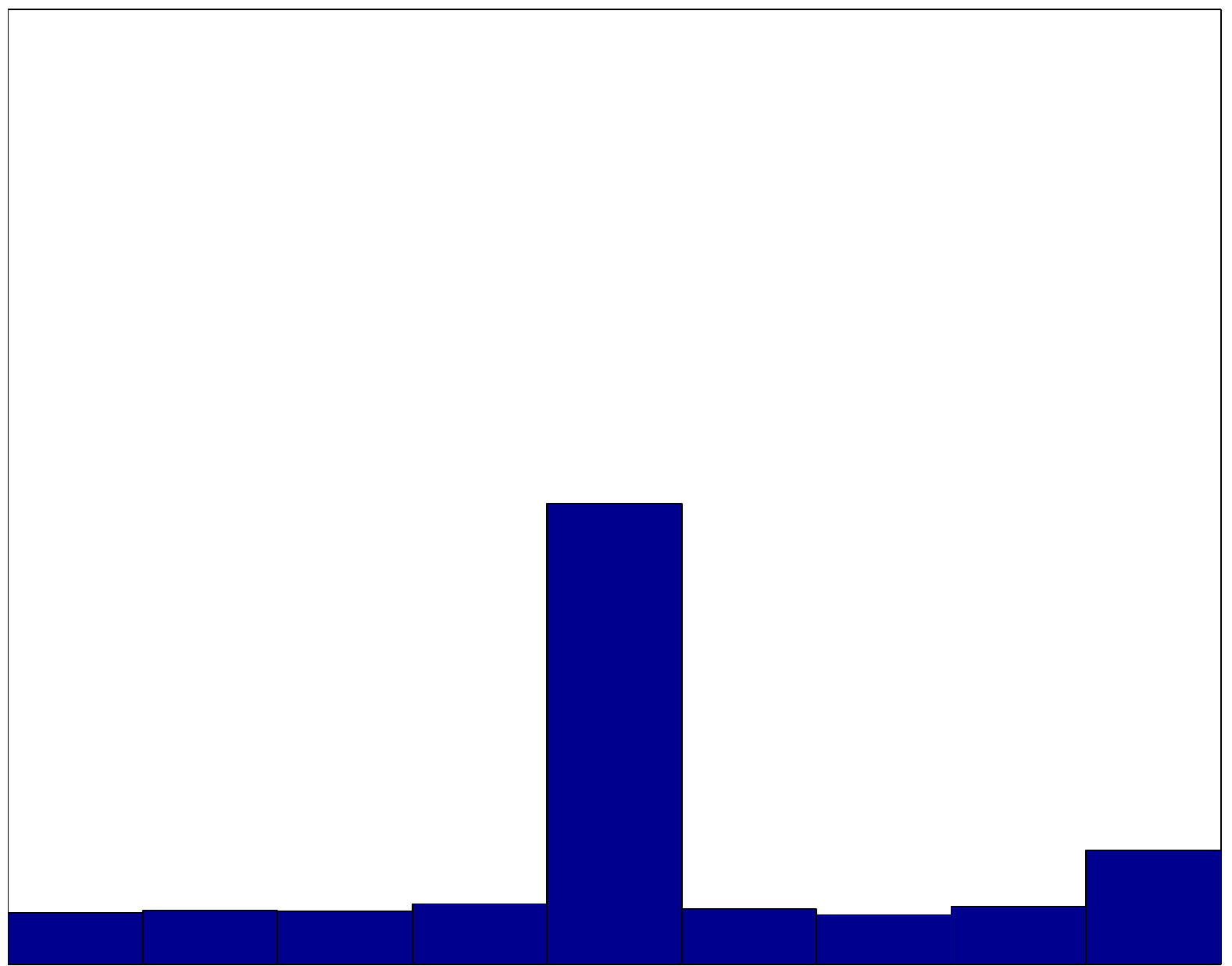}\includegraphics[width=\bb\columnwidth]{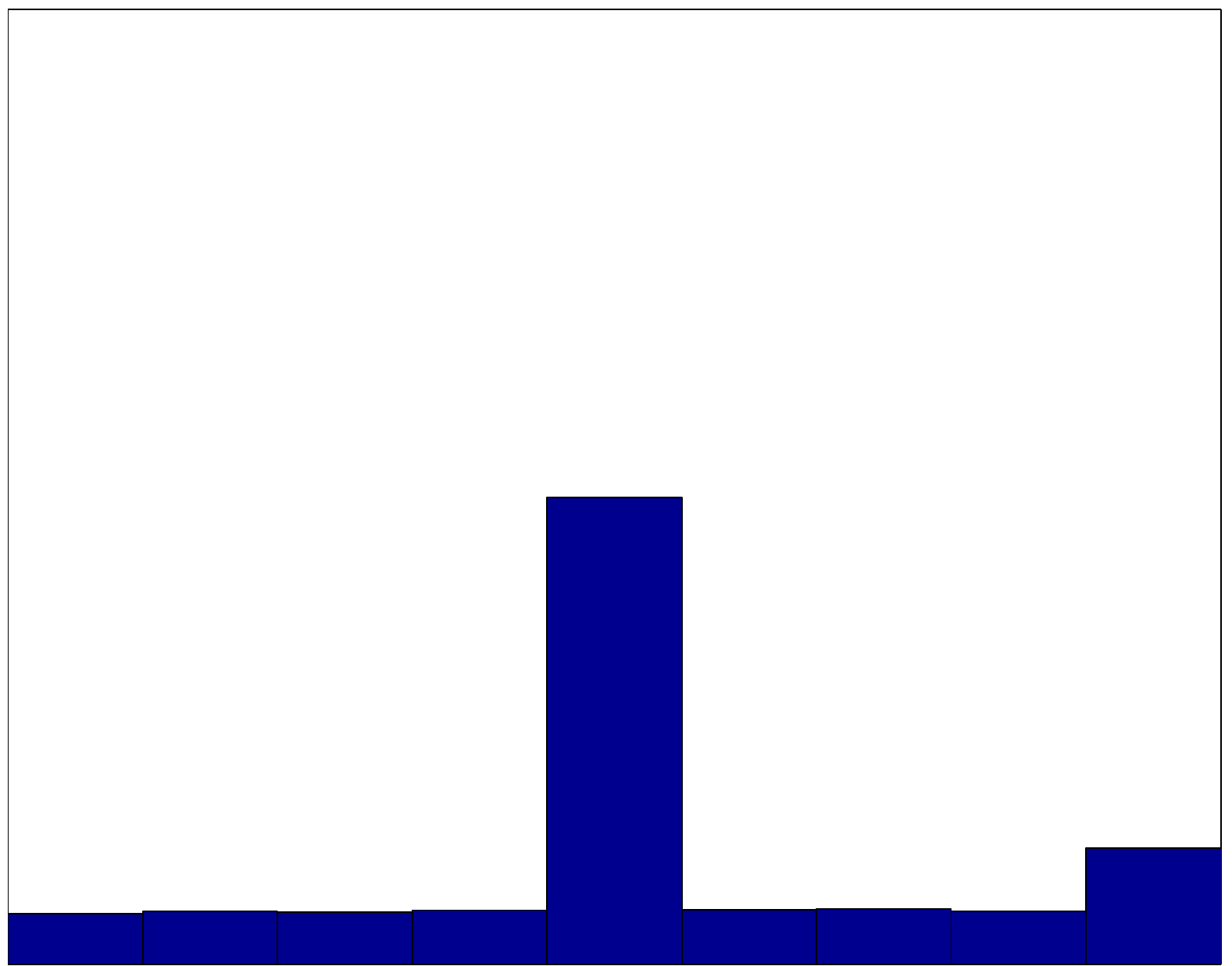}\caption{Histogram
  plots for the estimated shift and low SNR. From
left to right, top to bottom, $m=1,\dots,10$. The true shift was set
to 5 in all trials.}\label{fig:highSNR}\end{figure}

\begin{figure}[h!]
\includegraphics[width=\bb\columnwidth]{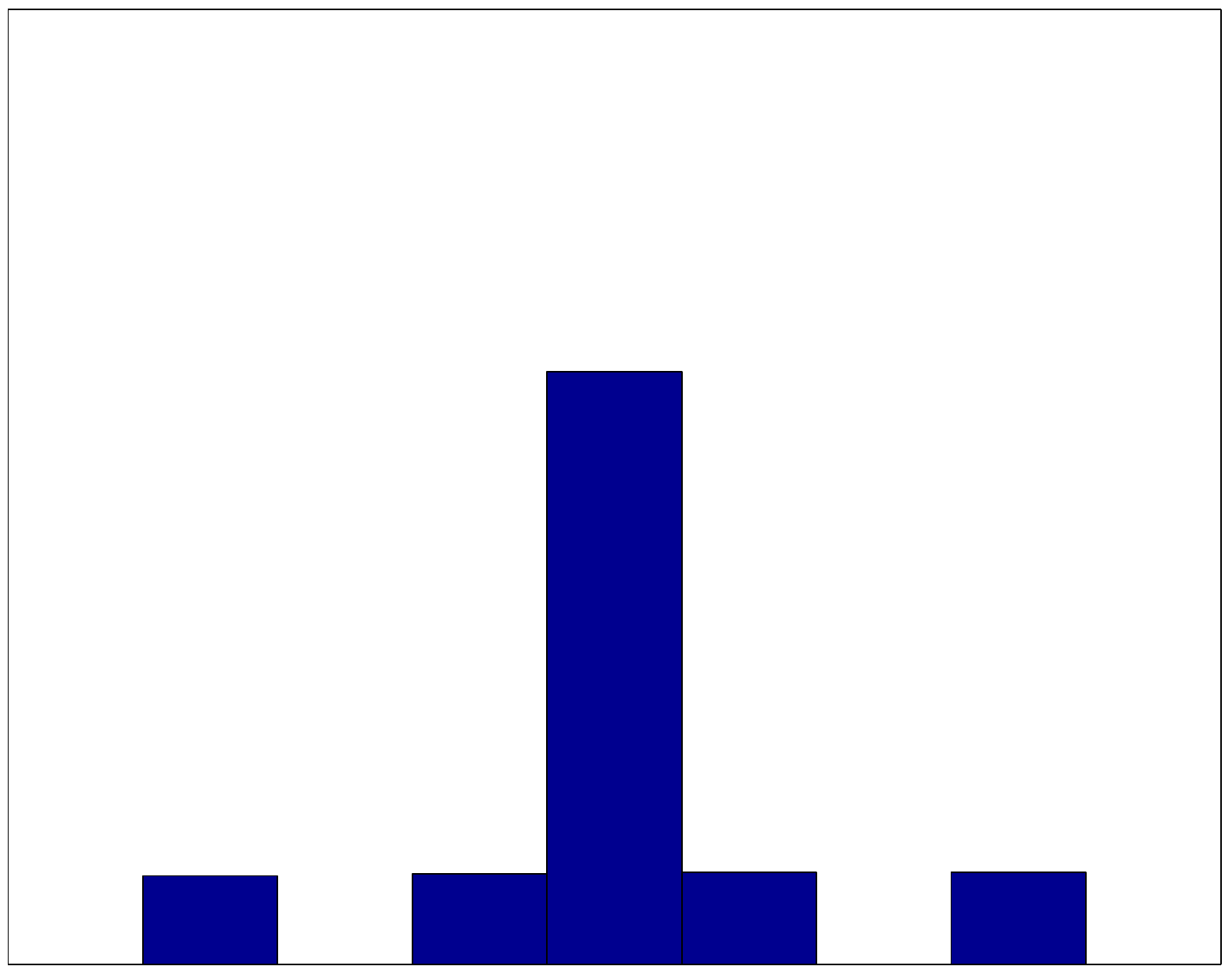}\includegraphics[width=\bb\columnwidth]{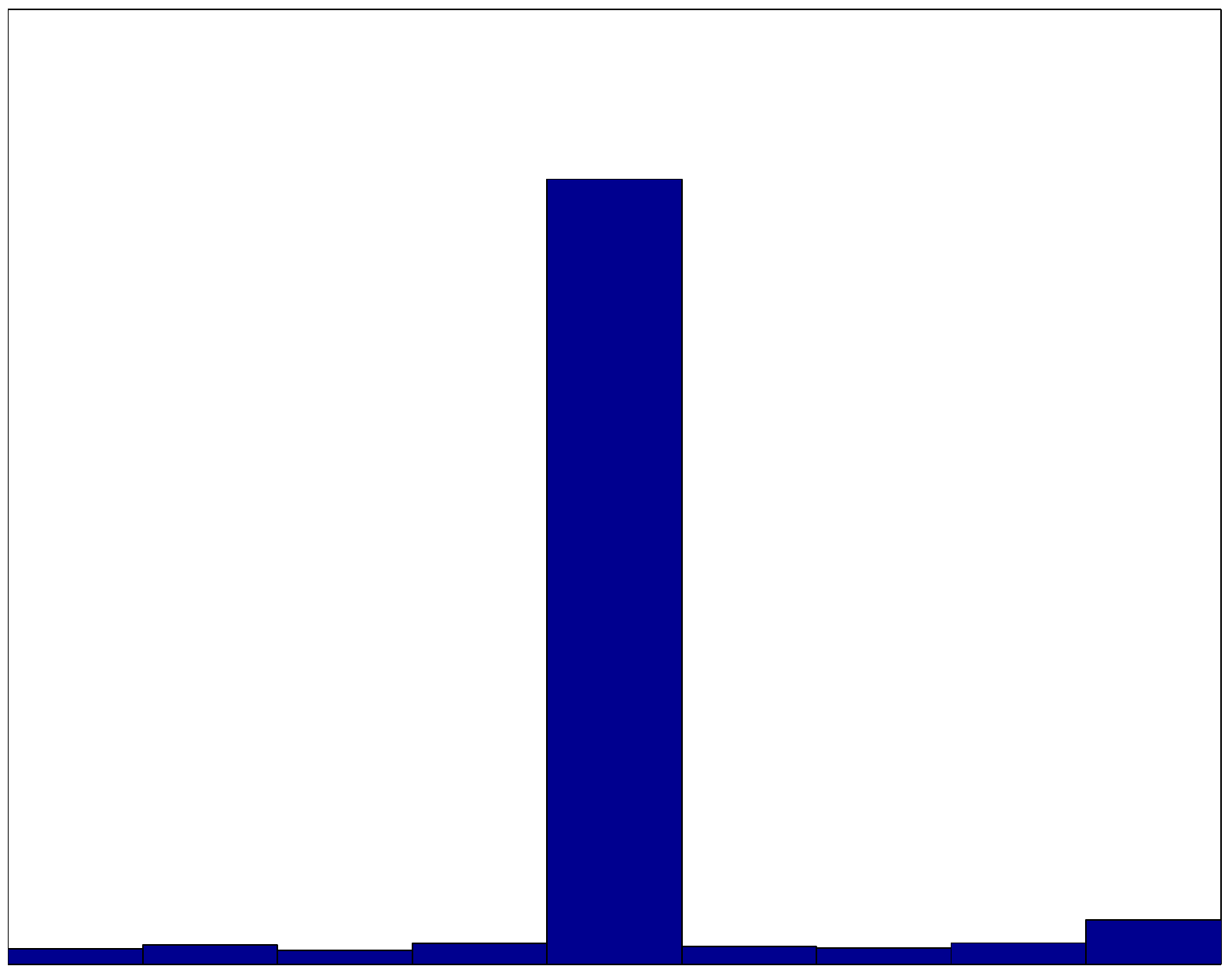}\includegraphics[width=\bb\columnwidth]{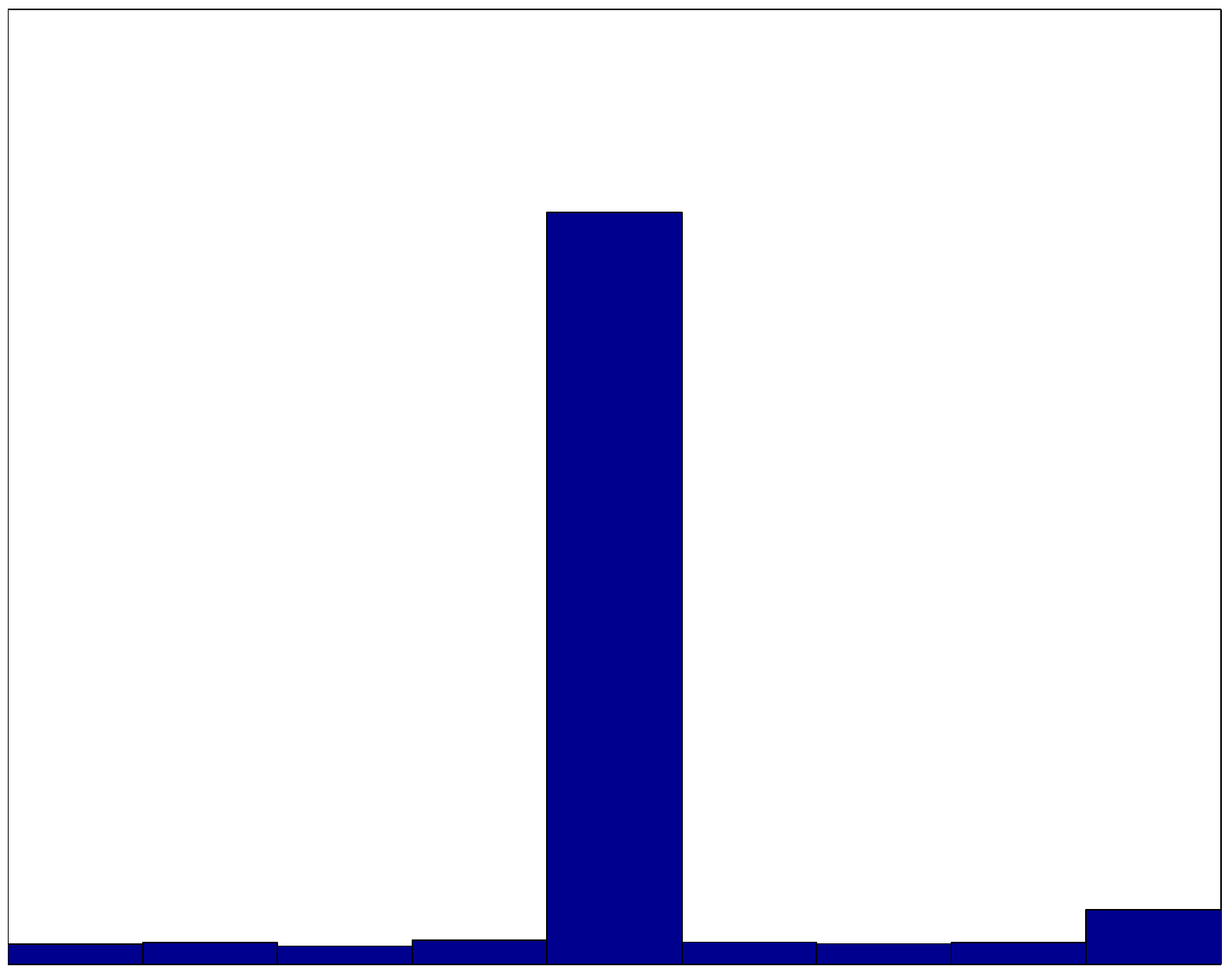}\includegraphics[width=\bb\columnwidth]{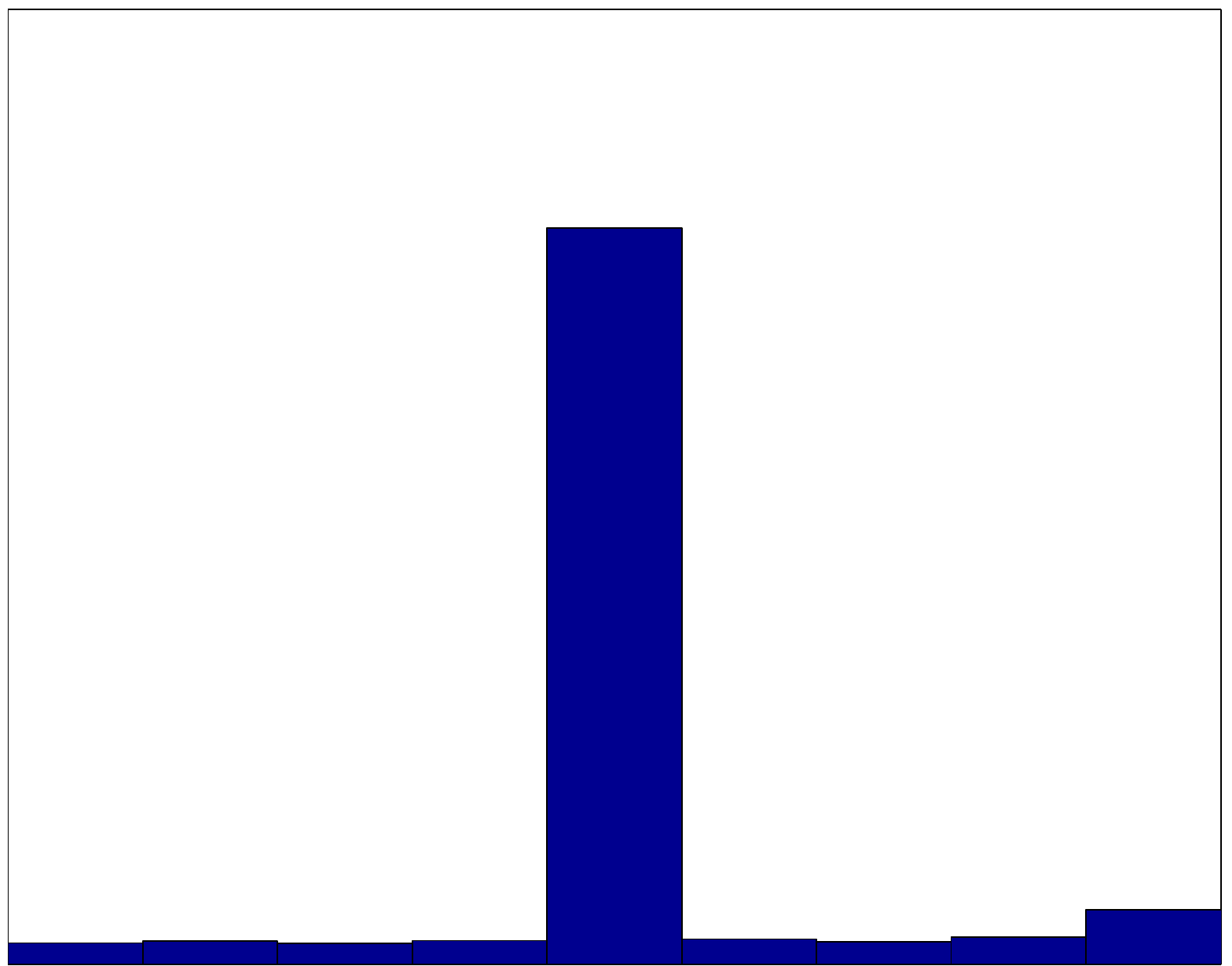}\includegraphics[width=\bb\columnwidth]{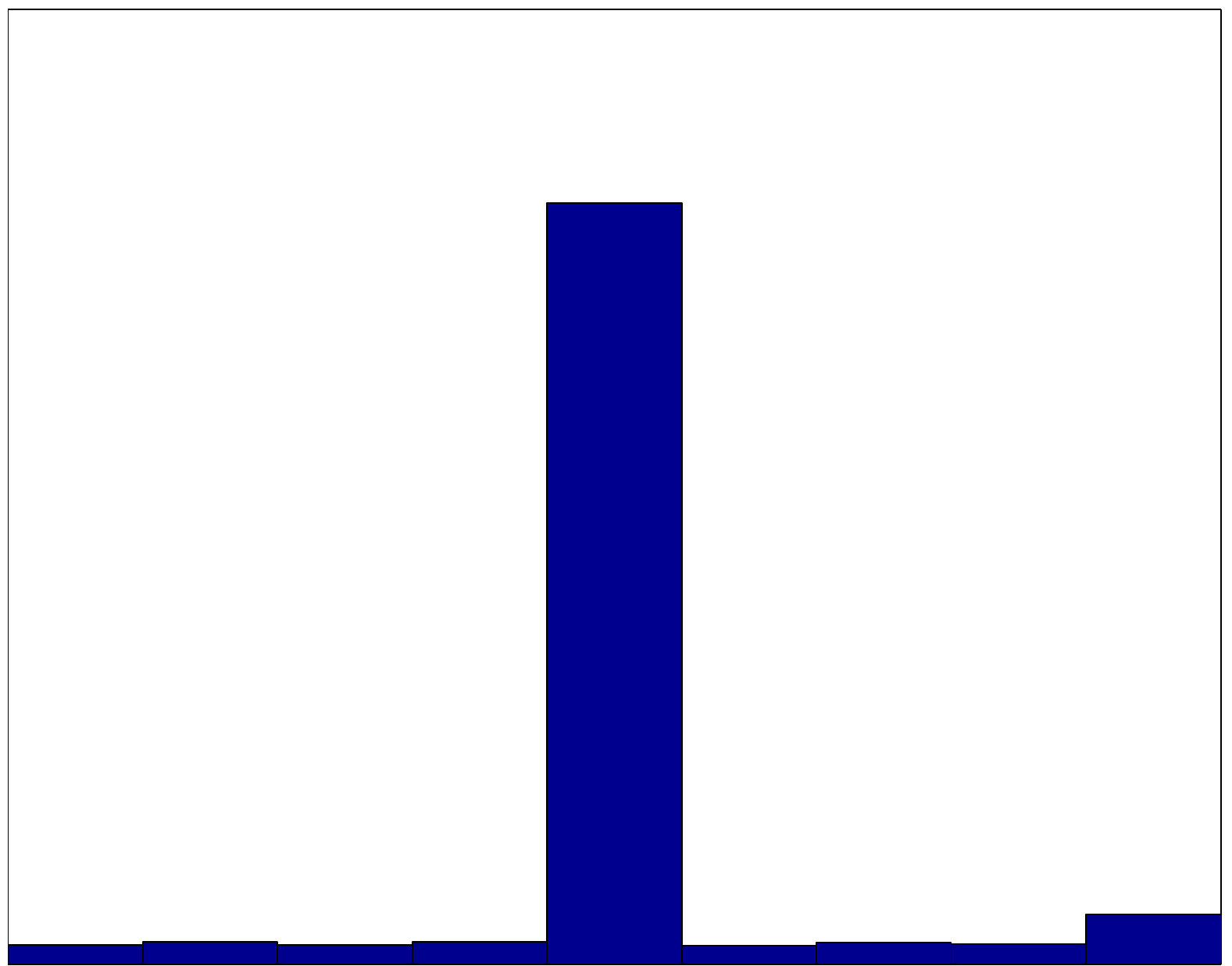}\\\includegraphics[width=\bb\columnwidth]{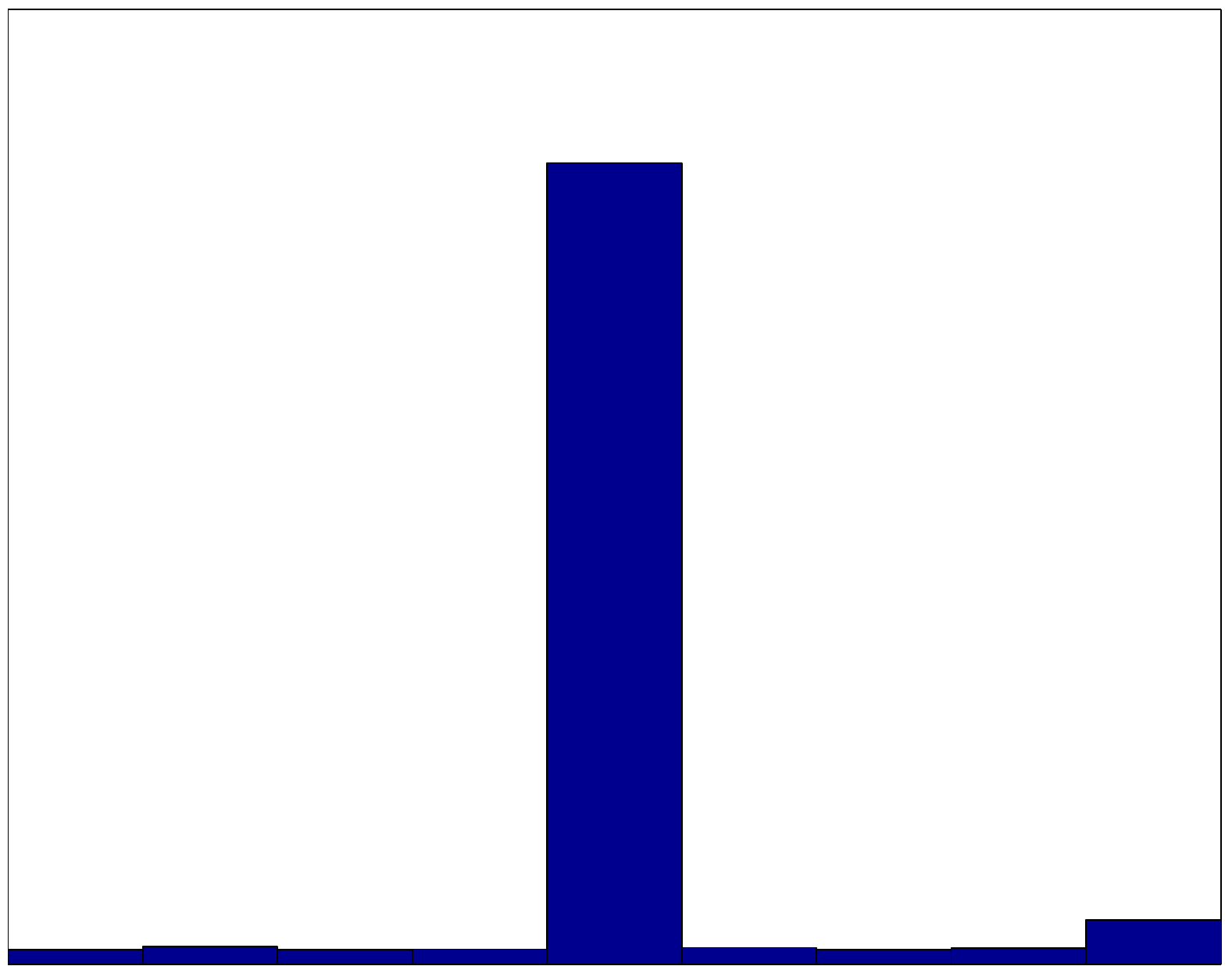}\includegraphics[width=\bb\columnwidth]{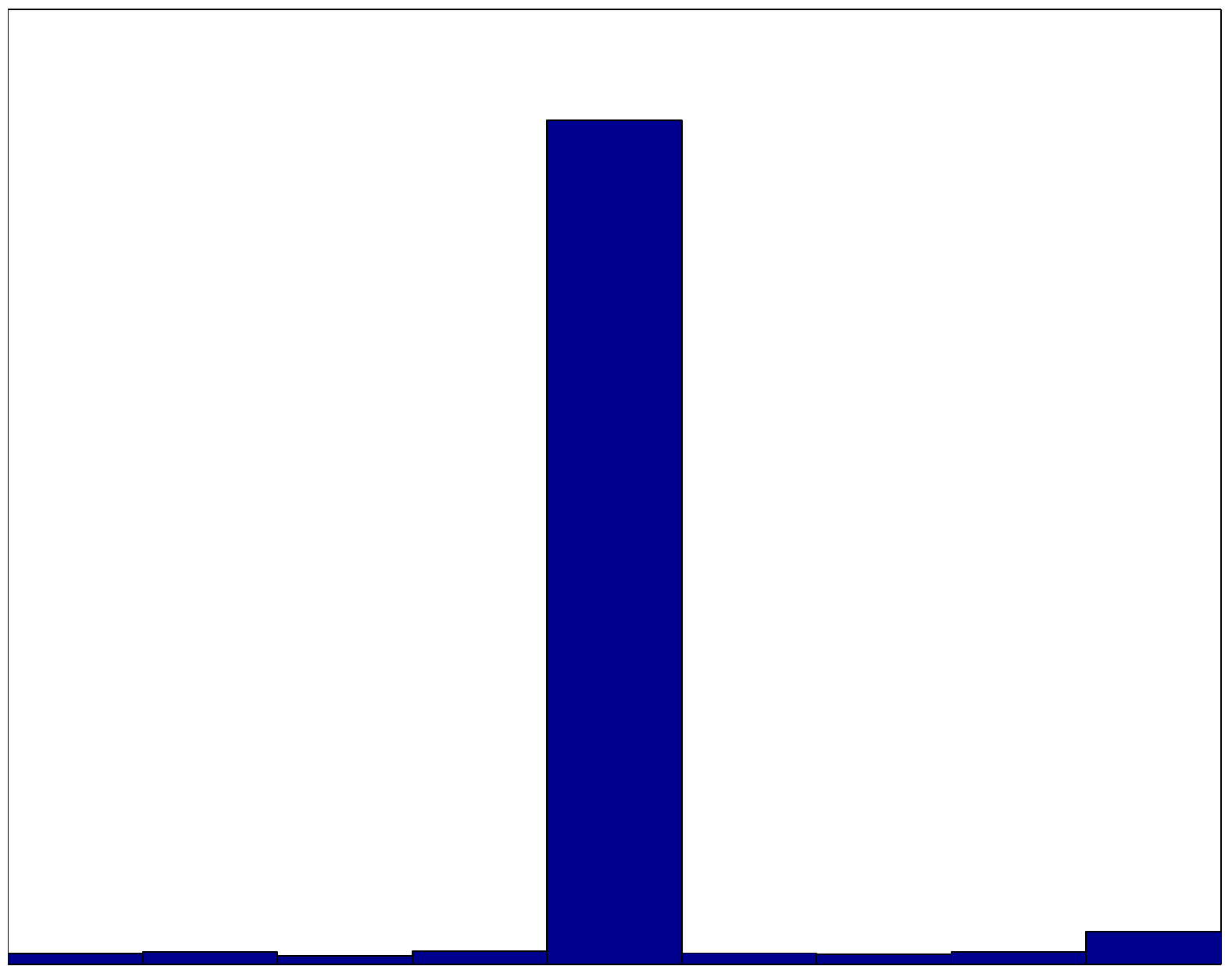}\includegraphics[width=\bb\columnwidth]{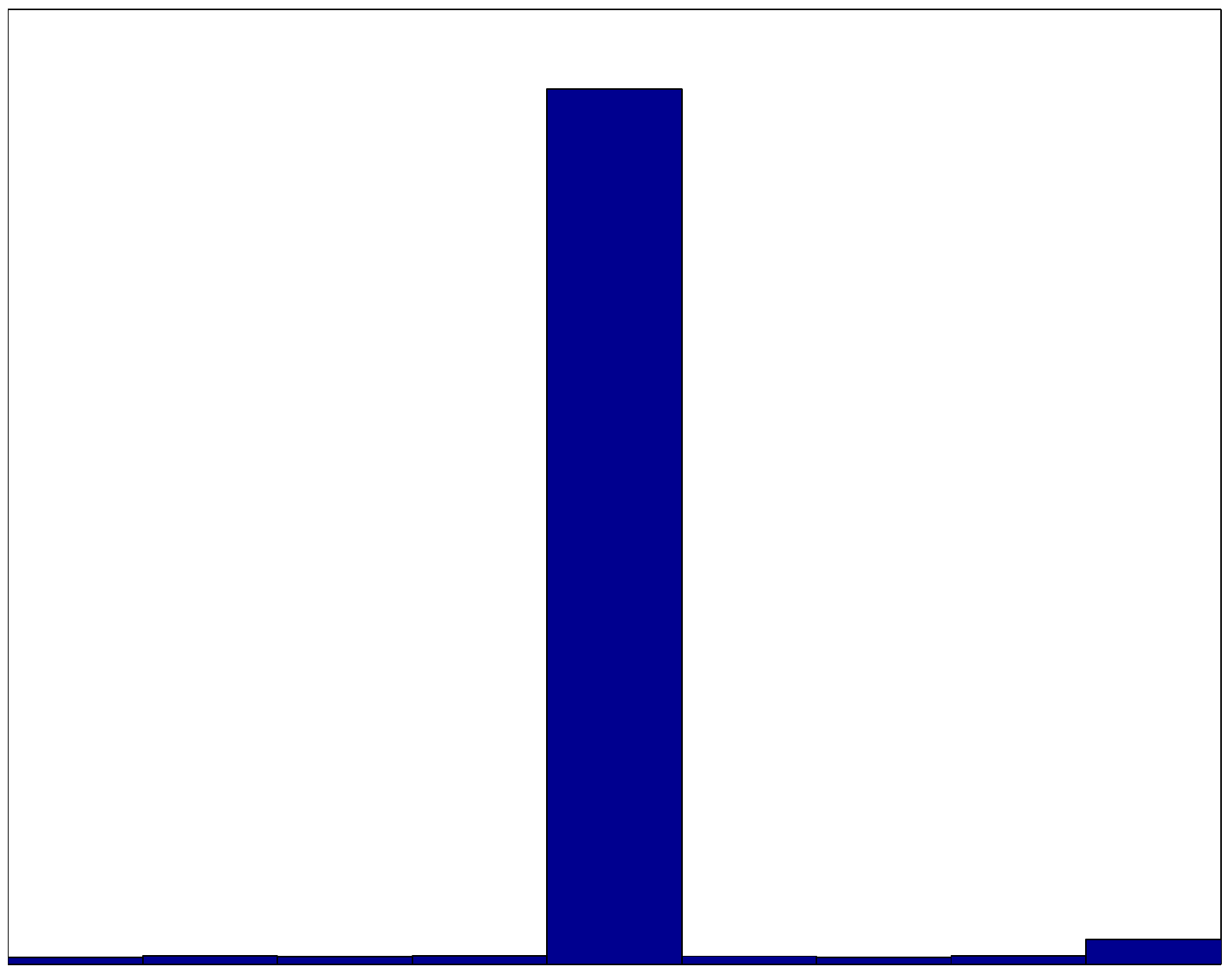}\includegraphics[width=\bb\columnwidth]{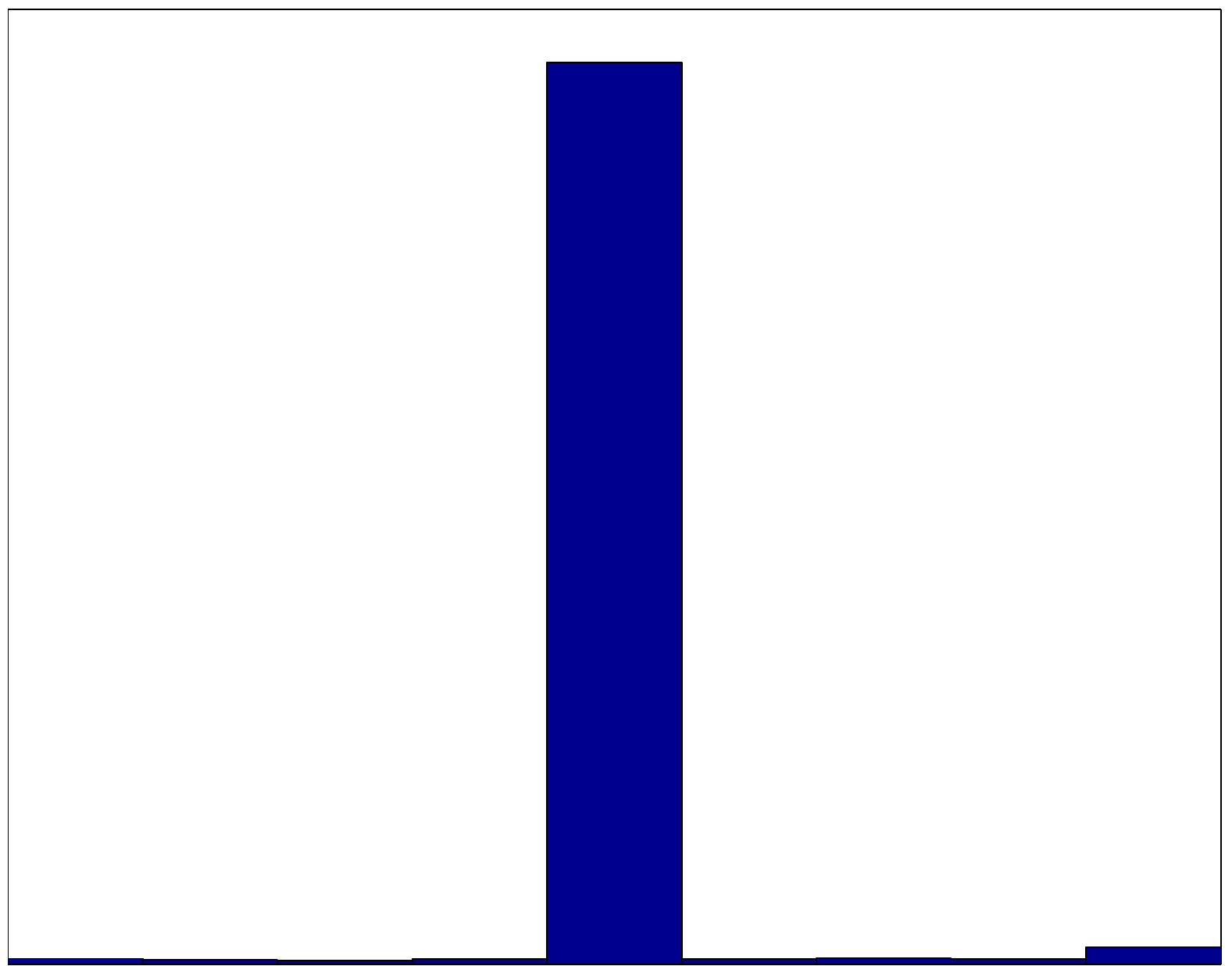}\includegraphics[width=\bb\columnwidth]{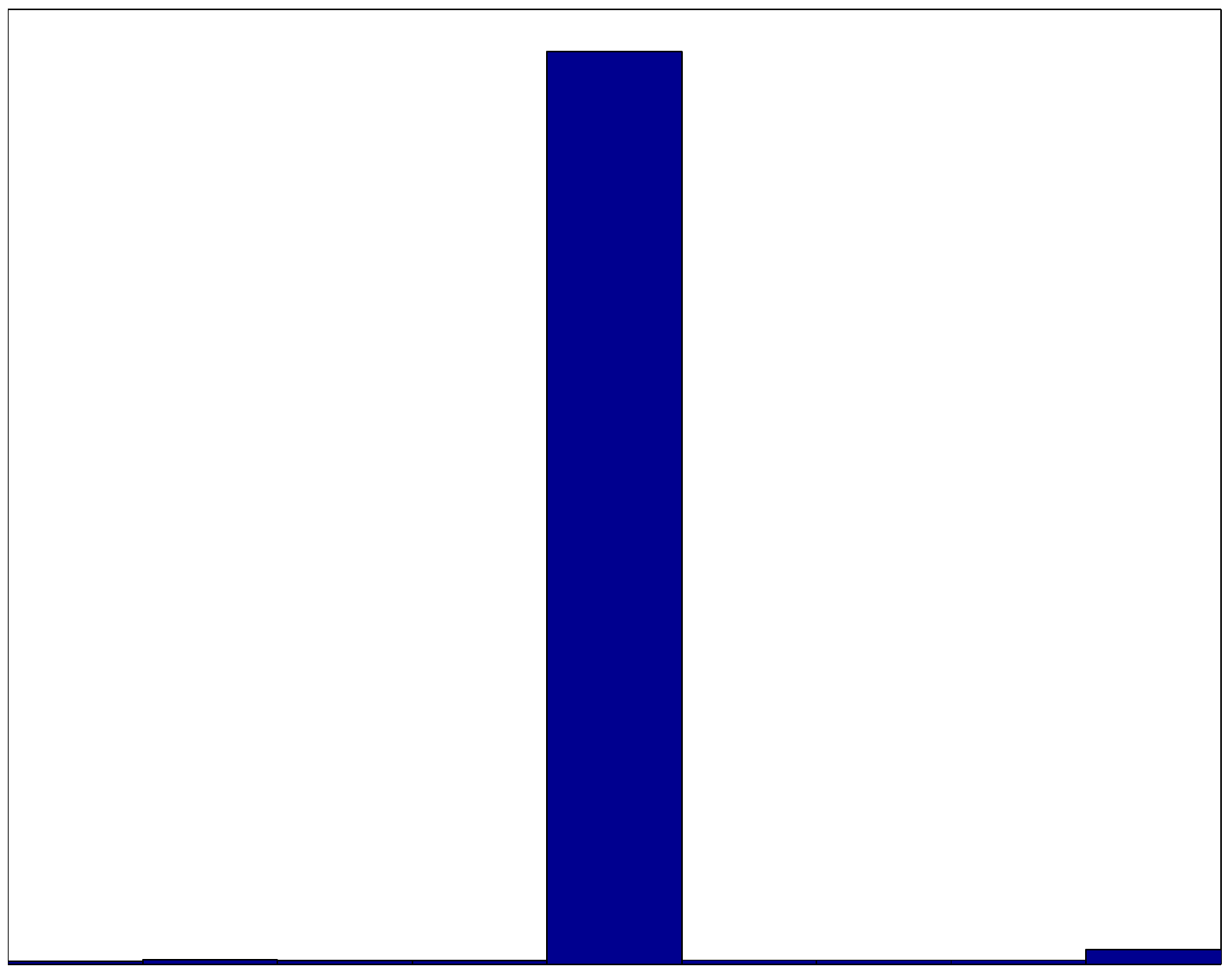}\caption{Histogram
  plots for the estimated shift and high SNR. From
left to right, top to bottom, $m=1,\dots,10$. The true shift was set
to 5 in all trials.}\label{fig:lowSNR}\end{figure}


Theorem \ref{thm:noisyrecov} gives conditions for when the noise does not affect the
estimate of the shift. This is a good property but even better would
be if the recovery of the true shift could be guaranteed. This is given by
the following corollary.

\begin{corr}[\bf Recovery of the True Shift from Noisy Low-Rate Data]\label{cor:truerecnoise}
If the $\ell_2$-norm difference between any two columns of $\AA  \tilde \XX $ is
greater than $2\|\ee_v \|_2$
and the conditions of Theorem \ref{thm:noisyrecov} are fulfilled, then
\eqref{eq:simptest} recovers the true shift.
\end{corr}


If the estimate of the shift has been computed,  a less conservative
test can be used to check if the computed estimate has been affected
by noise and if it is the true one. We summarize our conclusion in the
following corollary.
\begin{corr}[\bf Test for True Shift in the Presence of Noise]\label{corr:testnoise}
Assume that \eqref{eq:simptest}  gives  $s^*$ as an estimate of the
shift. If the $\ell_2$ difference between any  column and the
$s^*$-column of $\AA  \tilde \XX $ is
greater than $2\|\ee_v \|_2$ and $\Delta_{\zz\vv}$, then
$s^*$ is the true shift.
\end{corr}



\section{Conclusion}\label{sec:con}
To recover the cyclic shift relating two 1D signals,
the cross-correlation is usually evaluated for all possible shifts. Recent advances in hardware, signal acquisition and signal processing have made it possible to
sample or compute Fourier coefficients of a signal efficiently. It
is therefore of particular interest to see under what conditions the true shift
can be recovered from the Fourier coefficients. We have proposed a criterion that is computationally more efficient
than using the time samples, and we have shown that the true shift can be
recovered using as few as one Fourier coefficient. We have also derived
bounds for perfect recovery for both noise free and noisy
measurements. 



\appendices
\label{app:first}

\section{Proofs: Noise-Free Compressive Shift Retrieval}

Before proving the Theorem~\ref{thm:first}, we state two lemmas.
\begin{lem}[\bf Recovery of Shift using Projections]\label{lem:first}
Let $\XX$ be the $n \times n$-matrix made up of cyclically shifted versions
of $\xx$ as columns. If the columns of $\AA \XX$ are distinct, then the true shift can be recovered by
 \begin{equation}\label{eq:projected2}
\min_{\qq\in \{0,1\}^n}\| \AA \yy - \AA \XX \qq\|_2^2\quad  \text{s.t.} \quad \|\qq\|_0=1.
\end{equation}
\end{lem}
\begin{proof}[\bf Proof of Lemma \ref{lem:first}]
Since the shift relating $\xx$ and $\yy$ is assumed unique, it is
clear that the true shift is recovered by
\begin{equation}\label{eq:nonprojected2}
 \min_{\qq\in \{0,1\}^n}\|  \yy -  \XX \qq \|_2^2 \quad \text{s.t.} \quad \|\qq\|_0=1.
 \end{equation}
 Assume
that the solution of \eqref{eq:projected2} is not equivalent to that of
\eqref{eq:nonprojected2}. Namely, assume that \eqref{eq:nonprojected2}
gives $ \qq$, \eqref{eq:projected2} gives $\tilde \qq$ and $\qq\neq \tilde
\qq$. Since $ \qq$ will give a zero objective value in 
\eqref{eq:projected2}, so must $\tilde \qq$. We therefore have that
$\AA\yy=\AA \XX \tilde \qq=\AA \XX \qq$ and hence \begin{equation} 
\AA \XX \tilde \qq-\AA\XX \qq=\AA \XX (\tilde \qq- \qq) =0.
\end{equation}
Since $ \qq,\,\tilde \qq\in \{0,1\}^n$, $\|\tilde \qq\|_0=\|\qq\|_0=1$, and $\qq\neq \tilde
\qq$, $\AA \XX (\tilde \qq- \qq) =0$ implies that two
columns of $\AA \XX$ are identical. This is a contradiction and we
therefore conclude
that both \eqref{eq:projected2} and
\eqref{eq:nonprojected2}   recover the true shift. 
\end{proof}
\begin{lem}[\bf From \eqref{eq:projected2} to \eqref{eq:prodtest4}]\label{lem:two}
Under conditions 1) and 2) of Theorem \ref{thm:first},
the shifts recovered by \eqref{eq:projected2} 
and
\eqref{eq:prodtest4}
are the same.
\end{lem}
\begin{proof}[\bf Proof of Lemma \ref{lem:two}]
Consider the objective of  \eqref{eq:projected2}: 
\begin{align}\nonumber
\| \AA \yy - \AA \XX \qq\|_2^2=&(\AA \yy)^* \AA \yy +( \AA \XX \qq)^*
 \AA \XX \qq  \\ \label{eq:projected3} -&  (\AA \yy)^* \AA \XX \qq -(\AA \XX \qq)^*\AA \yy.
\end{align}
Writing $\XX\qq = \DD^s \xx$,  problem \eqref{eq:projected2} is
equal to
\begin{equation}\label{eq:prodtest3}
 \max_{s} 2\Re\{(\AA \yy)^* \AA  \DD^s \xx\} 
-( \AA  \DD^s \xx)^*
 \AA  \DD^s \xx. 
 \end{equation}
Now, if $\AA^* \AA \DD^s = \DD^s  \AA^* \AA $
and using that $( \DD^s)^* \DD^s=\II$ for a shift matrix, then 
\begin{align}
( \AA  \DD^s \xx)^*  \AA \DD^s \xx   =  \xx^*  (\DD^s)^* \AA^* \AA  
   \DD^s \xx=\|\AA \xx\|_2^2,
\end{align}
which is independent of $s$. Therefore, the shift recovered by
\eqref{eq:prodtest3} is the same as that of
\begin{equation}\label{eq:prodtest3222}
 \max_{s}\Re \{ (\AA \yy)^* \AA  \DD^s \xx\}. 
 \end{equation}
Lastly, if we again use  that $\AA^* \AA \DD^s = \DD^s  \AA^* \AA$ and
 $\alpha \AA \AA^*=\II$, then \eqref{eq:prodtest4} follows from 
\begin{align} 
\Re \{ (\AA \yy)^* \AA  \DD^s \xx\} &=\Re \{  \yy^*  \AA^* \AA  \DD^s \xx\} \\ &=
\alpha\Re \{  \yy^*  \AA^* \AA \AA^* \AA  \DD^s \xx\}
\\ &=
\alpha\Re \{  \yy^*  \AA^* \AA \DD^s  \AA^* \AA  \xx\} 
\\ &=
\alpha\Re \{ \langle  \zz, \bar \DD^s   \vv \rangle\} 
\end{align} where $\zz=\AA\yy$ and $\vv=\AA\xx$.
\end{proof}
We are now ready to prove Theorem \ref{thm:first}.
\begin{proof}[\bf Proof of Theorem \ref{thm:first}]
The assumptions of Theorem~\ref{thm:first} imply that requirements of both  Lemmas~\ref{lem:first}
and \ref{lem:two} are satisfied. The theorem therefore follows trivially. 
\end{proof}
We next prove Corollary \ref{cor:test}.
\begin{proof}[\bf Proof of Corollary \ref{cor:test}]
In the proof of Lemma \ref{lem:first}, $\AA \XX (\tilde \qq- \qq) =0$
leads to $\tilde \qq- \qq =0$ if the columns of $\AA \XX $ were all
distinct. Now, if 
\begin{equation}\label{eq:prodtest32222}
 s^*=\argmax_{s}\Re \{ \langle  \zz, \bar \DD^s   \vv \rangle\},  
 \end{equation}
that corresponds to the $s^*$th   element of $\tilde \qq$ being one
and all other elements zero. Hence, Lemma \ref{lem:first} can be made less
conservative if $s^*$ is known by requiring that only the $s^*$th   column of  $\AA
\XX$ is different than all other columns.
\end{proof}
\begin{proof}[\bf Proof of Lemma \ref{lem:com}]
Let $\MM = \AA \DD^s$ and $\QQ = \AA (\DD^s)^*$. By the definition of $\DD^s$, $\MM$ is a column
permutation of $\AA$ where the columns are shifted $s$ times to the
right. Thus, the $r$th column of $\MM$ is equal to the $t$th column of
$\AA$ where $t=(r-s) \text{ mod } n$. It is also easy to see that
$(\DD^s)^*$ permutes the columns of $\AA$ by $s$ to the left so that the
$r$th column of $\QQ$ is equal to the $q$th column of $\AA$ where $q=(r+s)
\text{ mod } n$. Now, the $pr$th element of $\AA^* \MM = \AA^* \AA\DD^s$ is given by
\begin{equation}
(\AA_{:,p})^*  \MM_{:,r} = (\AA_{:,p})^* \AA_{:,r-s} = \frac{1}{n} \sum_{i=1}^m e^{2j\pi k_i (p-r+s)},
\end{equation}
where $\AA_{:,p}$ is used to denote the $p$th column of $\AA$ and $  \MM_{:,r}
$ the $r$th column or $\MM$. On the other hand, the $(p,r)$-th element of $\QQ^* \AA = \DD^s \AA^* \AA$ is given by
\begin{equation}
(\QQ_{:,p})^*  \AA_{:,r} = (\AA_{:,p+s})^* \AA_{:,r} = \frac{1}{n} \sum_{i=1}^m e^{2j\pi k_i (p+s-r) }.
\end{equation}
Clearly, the two are equivalent.
\end{proof}
We are now ready to prove Corollary \ref{cor:first}.
\begin{proof}[\bf Proof of Corollary \ref{cor:first}]
Lemma \ref{lem:com} gives that Condition 1) of Theorem \ref{thm:first}  is
satisfied. Since a full Fourier matrix is orthonormal, a matrix made
up of a selection of
rows of a Fourier matrix satisfies Condition 2). 
The last condition of Theorem \ref{thm:first} requires columns of
$\AA\XX$ to be distinct. A sufficient condition is that there exists a
row with all distinct elements.  
As shown previously, the $(p,r)$-th element of $\AA\XX$ is $X_{k_p}
e^{\frac{2 j \pi k_p (r-1)}{n} }.$ If $X_{k_p}$ is assumed nonzero,
a sufficient condition for $\AA\XX$ to have distinct columns is that 
 $e^{\frac{2 j \pi k_p r_1}{n} } \neq e^{\frac{2 j \pi k_p
     r_2}{n} },\, r_1,r_2 \in \{ 0,\dots,n-1\}, r_1\neq r_2$.
 This condition can be simplified to $\frac{ k_p r_1}{n}
 \neq \frac{ k_p  r_2}{n}  + \gamma,\,\gamma\in \Z$. By realizing that $r_1-r_2$ takes values in $\{
 -n+1,\dots,-1,1,\dots,n-1\}$ we get that the condition is equivalent
 to requiring that there is no integers in $\{
 -n+1,\dots,-1,1,\dots,n-1\}\frac{k_p}{n}$.  Due to symmetry, a
 sufficient condition for distinct columns is that there exists a $p\in \{1,\dots,m \}$ such
that $X_{k_p}\neq 0$ and $\{1,\dots,n-1 \} \frac{k_p}{n}$ contains no integers. 
Lastly, if we write out $\AA  \DD^s\AA ^* $ we get that the $pr$th element is equal to  $\delta_{p,r} e^{ -\frac{2 j\pi
  k_p  s}{n}} /n$ and hence  the simplified  test proposed in \eqref{eq:simptest}. 
\end{proof}

\section{Proofs: Noisy Compressive Shift Retrieval}

 \begin{proof}[\bf Proof of Theorem \ref{thm:noisyrecov}]
 From Lemma \ref{lem:two} we can see that seeking $s$ that maximizes
 $\Re \{ \langle \tilde \zz, \bar \DD^s \tilde \vv  \rangle \}$ is equivalent to
 seeking $\qq$ that solves 
 \begin{equation}\label{eq:l1crit}
 \min_{\qq\in \{0,1\}^n}\| \tilde  \zz - \AA  \tilde \XX\qq \|_2^2\quad  \text{s.t.} \quad \|\qq\|_0=1,
 \end{equation}
 where the first column of $\AA  \tilde \XX $ is equal to $\tilde \vv$ (which
 defines the first column of $\tilde \XX $) and  
the $i$th
 column of  $\tilde \XX $   a circular shift of the first column $i-1$
 steps. Assume that $\hat \qq$ solves \eqref{eq:l1crit}.  Since our measurements are noisy, we can not
 expect a zero loss. The loss can be shown given by
 \begin{equation}
 \| \tilde \zz -\AA  \tilde \XX\hat \qq\|_2^2 =\|\tilde \vv\|_2^2+\|\tilde \zz\|_2^2- \max_s2\Re \{\tilde \zz^* \bar
 \DD^s \tilde \vv \}.
 \end{equation}  
 Now, consider $\| \tilde  \zz - \AA  \tilde \XX\hat \qq \|_2$. Assume that $
 \qq_0$ solves the noise-free version of \eqref{eq:l1crit}  
and let $\tilde \XX =\XX +\HH$. We have the
 following inequality:
 \begin{align}
\nonumber 
 \| \tilde  \zz - \AA  \tilde \XX\hat \qq \|_2 &= \| \zz+\ee_z-\zz+\AA  \XX \qq_0 - \AA  \tilde \XX 
 \hat \qq \|_2 \\ \nonumber &= \| \ee_z+\AA  \XX \qq_0 - \AA  \tilde \XX 
 \hat \qq \|_2 \\ \nonumber  &= \| \ee_z+\AA  (\tilde \XX -\HH) \qq_0 - \AA  \tilde \XX 
 \hat \qq \|_2 \\ \nonumber &\geq  \|\AA  \tilde \XX \qq_0 - \AA  \tilde \XX 
 \hat \qq \|_2-\| \ee_z\|_2-\| \ee_v \|_2,
 \end{align}
where we used the fact that $\AA  \HH \qq_0=\ee_v$. Therefore
\begin{align} 
 \|\AA  \tilde \XX \qq_0 -& \AA  \tilde \XX 
\hat \qq \|_2  \leq \Delta_{\zz\vv}.
\end{align}
Since  $\|\hat \qq\|_0=\|\qq_0\|_0=1$, if the 
$\ell_2$ difference between any two columns of $\AA  \tilde \XX $ is greater
than $\Delta_{\vv\zz}$, then $\qq_0=\hat \qq$. 
\end{proof}


\begin{proof}[\bf Proof of Corollary \ref{cor:truerecnoise}]
Let $\tilde \qq$ and $\hat \qq$ be any vectors such that $\|\hat
\qq\|_0 =\|\tilde \qq\|_0 =1$, $\hat  \qq \neq \tilde  \qq$ and $\hat  \qq,\tilde \qq
\in \{0,1\}^n$. Using the triangle inequality we have that
\begin{align}
\|\AA \tilde \XX \hat  \qq& -\AA \tilde \XX \tilde \qq \|_2= \|\AA (\XX +\HH) (\hat  \qq - \tilde \qq) \|_2 \\
\leq &
\|\AA \XX (\hat  \qq - \tilde \qq) \|_2+\|\AA  \HH (\hat \qq - \tilde \qq) \|_2  \\ \leq & \|\AA \XX (\hat  \qq -\tilde \qq) \|_2+2\|\ee_v \|_2.
\end{align}
Hence, if $\|\AA  \tilde \XX \hat \qq -\AA \tilde \XX \tilde \qq \|_2-2\|\ee_v \|_2> 0$,
then $ \|\AA \XX (\hat  \qq - \tilde \qq) \|_2$ is greater than zero. Now
since Theorem \ref{thm:noisyrecov} gives that  \eqref{eq:simptest}
recovers the same shift as if the measurements would have been
noise-free, and since Theorem~\ref{thm:first} gives that the
noise-free estimate is equal to the true shift if $ \|\AA \XX (\hat
\qq - \tilde \qq) \|_2$ is greater than zero (or equivalent that all
columns of $ \AA \XX $ are distinct), we can guarantee the
recovery of the true shift also in the noisy case.  
\end{proof}

\begin{proof}[\bf Proof of Corollary \ref{corr:testnoise}]
The corollary follows trivially by setting the $s^*$th  element of
$\hat \qq$ to one
and all other elements zero in the proofs of Theorem
\ref{thm:noisyrecov} and Corollary \ref{cor:truerecnoise}. 
\end{proof}






\ifCLASSOPTIONcaptionsoff
  \newpage
\fi

\bibliographystyle{IEEEtran}
\bibliography{paper}
\end{document}